\documentclass[ba,preprint]{imsart}
\pubyear{2021}
\volume{TBA}
\issue{TBA}
\firstpage{1}
\lastpage{1}

\usepackage{amsthm}
\usepackage{amsmath}
\usepackage{amsfonts}
\usepackage{amssymb}
\usepackage{natbib}
\usepackage[colorlinks,citecolor=blue,urlcolor=blue,filecolor=blue,backref=page]{hyperref}
\usepackage{graphicx}
\usepackage{float} 
\usepackage{booktabs} 
\usepackage{graphicx} 
\usepackage[margin=1cm]{caption} 
\usepackage{mathtools}
\usepackage{subfigure}
\usepackage{pgfplots}
\usepackage{tikz}
\newlength\figureheight
\newlength\figurewidth
\usepackage{algorithm}
\usepackage{algorithmic}

\floatstyle{plain}
\newfloat{Algorithm}{thp}{lop}
\floatname{Algorithm}{Algorithm}

\def\1{1\!{\rm l}}

\theoremstyle{definition}

\newtheorem{assumption}{Assumption}

\newtheorem{corollary}{Corollary}
\newtheorem{lemma}{Lemma}

\theoremstyle{definition}
\newtheorem*{definition}{Definition}

\newcommand{\dt}{\mathrm{d}}
\newcommand{\tr}{\mathrm{tr}}

\newcommand{\I}{\mathbb{I}}

\newcommand{\E}{\mathbb{E}}

\newcommand{\bz}{\boldsymbol{z}}
\newcommand{\by}{\boldsymbol{y}}

\newcommand{\y}{\mathbf{y}}
\newcommand{\Y}{\mathbf{Y}}
\newcommand{\z}{\mathbf{z}}
\newcommand{\Z}{\mathbf{Z}}

\def \E{\mathbb{E}}

\def \dt {\mathrm{d}}

\startlocaldefs
\numberwithin{equation}{section}
\theoremstyle{plain}

\endlocaldefs

\begin{document}

\begin{frontmatter}
\title{Pooling information in likelihood-free inference\support{This work was supported by the Australian Research Council and a Singapore	Ministry of Education Academic Research Fund Tier 1 grant.}}
\runtitle{Pooling information in LFI}

\begin{aug}
\author{\fnms{David T.} \snm{Frazier}\thanksref{addr1}\ead[label=e1]{david.frazier@monash.edu}},
\author{\fnms{Christopher} \snm{Drovandi}\thanksref{addr2}\ead[label=e2]{c.drovandi@qut.edu.au}},
\author{\fnms{Lucas} \snm{Kock}\thanksref{addr3}\ead[label=e3]{lucas.kock@nus.edu.sg}}
\and
\author{\fnms{David J} \snm{Nott}\thanksref{addr3}\ead[label=e4]{standj@nus.edu.sg}}

\runauthor{D.T. Frazier, C. Drovandi, L. Kock and D.J. Nott}

\address[addr1]{Department of Econometrics and Business Statistics, Monash University, Australia.
    \printead{e1} 
}
\address[addr2]{School of Mathematical Sciences, Queensland University of Technology, Brisbane 4000 Australia 
	\printead{e2}
}

\address[addr3]{Department of Statistics and Data Science, National University of Singapore, Singapore. 
	\printead{e3,e4}
}
\end{aug}

\begin{abstract}
		Likelihood-free inference (LFI) methods, such as approximate Bayesian computation, have become commonplace for conducting inference in complex models. Many approaches are based on summary statistics or discrepancies derived from synthetic data. However, determining which summary statistics or discrepancies to use for constructing the posterior remains a challenging question, both practically and theoretically. Instead of relying on a single vector of summaries for inference, we propose a new pooled posterior that optimally combines inferences from multiple LFI posteriors. This pooled approach eliminates the need to select a single vector of summaries or even a specific LFI algorithm.  Our approach is straightforward to implement and avoids performing a high-dimensional LFI analysis involving all summary statistics. We give theoretical guarantees for the improved performance of the pooled posterior mean in terms of asymptotic frequentist risk and demonstrate the effectiveness of the approach in a number of benchmark examples.
\end{abstract}

\end{frontmatter}

	\section{Introduction}\label{sec:introduction}

The complexity of many models encountered in modern applications has led to 
the development of new inferential methods which are applicable 
when the likelihood function
is intractable.  To perform Bayesian inference with intractable 
likelihoods so-called likelihood free inference (LFI) methods are 
commonly used, which 
replace likelihood evaluations by model simulations.  One of the most
well-established LFI methods is approximate Bayesian computation (ABC); for a review of ABC see the handbook \cite{sisson2018handbook}. 

LFI assumes that the observed data is drawn from a given class of models from which it is feasible to generate synthetic data. Common LFI methods construct an approximate posterior for the model unknowns by comparing, in a given distance, summary statistics calculated using the observed data and data simulated from the model. This approach permits statistical inference in  complex models, but with accuracy heavily depending on the choice of summary statistics. 

Different choices of summaries result in different posteriors, and can sometimes produce surprisingly disparate inferences. Directly comparing different
collections of summary statistics is not easy:
{under regularity conditions, asymptotically the LFI posterior variance can be shown to be (weakly) decreasing in the number of summaries used in the analysis (see \citealp{FMRR2016} and \cite{LF2016} for details).}
{{However, we acknowledge that such an asymptotic viewpoint disregards finite-sample differences in the locations and scales of posteriors that can result from employing different collections of summaries.} 
	In practice, adding more summaries also increases the computational burden. Even if additional summaries are highly informative, adequately controlling the additional Monte Carlo error resulting from their inclusion may not be possible with the available computational resources.

	In this paper, we make three contributions to the literature on LFI. Firstly, rather than choosing summary statistics, we propose to conduct LFI by combining several posteriors built using different summary statistic vectors. While it may be possible to fuse posteriors in many different ways, our suggested approach uses linear opinion pools  (\citealp{stone1961opinion}), due in part to their simplicity and good performance in many tasks (e.g.~\citealp{McaRei2022}, \citealp{AriTunBenBudDieHonWalZau2000}). Linear opinion pools are known to be useful tools for combining prior beliefs or evidence.  We refer to \cite{evans2022combining} for a recent discussion of the latter application in likelihood-based Bayesian inference.   The linear pooling approach is computationally attractive, since it allows us to efficiently combine the information from many different sets of summary statistics without requiring a high-dimensional LFI analysis considering all of them simultaneously (see, e.g., \citealp{blum2010non} for a discussion of the curse of dimensionality in LFI). 	
		As well as simple linear opinion pools, we also consider a variant where
		mixture components in the pool are recentred which avoids variance 
		inflation when combining posteriors with very different locations.  
		The theory we develop for our pooling method applies in both cases.
		
		Secondly, we show that the pooling approach can be applied to combine inferences from summary-based LFI  posteriors and those built using general discrepancy measures, such as the Wasserstein distance (\citealp{Bernton2017}), the energy distance (\citealp{NguArbJulLueFor2020}), or the Kullback-Leibler divergence (\citealp{Jia2018}); see \citet{drovandi2021comparison} for a review of such approaches in LFI. Such a combination has not been considered previously to the best of
		our knowledge. It is also possible to pool inferences from different 
		summary statistic based LFI algorithms, such as ABC and Bayesian synthetic likelihood (BSL, \citealp{price2018bayesian}).
		
		Lastly, we show theoretically that, under certain assumptions, the pooled posterior mean has improved performance for point estimation compared to the posterior mean for any individual collections of summaries, in terms of asymptotic frequentist risk. In addition, we show that in cases where one set of summaries is incompatible with the assumed model (see, e.g., \citealp{marinea2012}, \citealp{frazier2020model} or Section \ref{sec:asymp} for discussion), the pooled posterior automatically disregards the incompatible set of summaries. 	
		In principle, the theory developed also applies to the more general case of combining summary-based LFI posteriors and those based on general
		discrepancy measures.  However, a rigorous extension to that case would require asymptotic normality of the posterior mean from the ``discrepancy-based" posterior, which has not been theoretically verified at present, and a formal
		analysis of this and possible extensions to simulator-based inference
		(e.g.\ \citealp{Tejero-Cantero2020}) is left to future research.  
		
		The remainder of the paper proceeds as follows. Section~\ref{sec:generativemodels} contains the motivation and general setup. Section~\ref{sec:optpost} provides the intuition for the pooling approach, along with a na\"{i}ve implementation method, and some
		illustrative examples. Theoretical aspects of the pooling approach are also discussed.  Section~\ref{sec:poolBSL+ABC} extends the pooling approach to the case of general discrepancy based measures, and demonstrates the appreciable inferential gains that can be obtained in this setting. Section~\ref{sec:discussion} concludes with a discussion on future work. Supplementary material for this paper includes: additional discussion and examples (Section~A),  as well as proofs of all results stated in the main text (Section~B).
		
		\section{Likelihood-free inference and the choice of summaries\label{sec:generativemodels}}
		\subsection{Likelihood-free inference}
		For a sample size $n\ge1$, let $(\Omega_n,\mathcal{F}_n,\mathbb{P}_n)$ denote the  probability space, with associated expectation operator $\mathbb{E}_n$, on which all random variables are defined. For simplicity of notation, we drop quantities dependence on $n$ when no confusion will result. Denote by $\mathcal{P}(\mathcal{X})$ the set of probability measures on a space $\mathcal{X}$. We observe data $\by = (y_1,\ldots,y_n)^\top\in \mathcal{Y}^n$,  distributed according to some unknown measure $P_0^{(n)}$.

		Our beliefs about $P_0^{(n)}$ are specified as a class of parametric models $\mathcal{M}^{(n)}=\{P^{(n)}_\theta : \theta \in \Theta\} \subseteq
		\mathcal{P}(\mathcal{Y}^n),\text{ where }\Theta\subseteq \mathbb{R}^{d_\theta}.$ We quantify our prior beliefs about $\theta$ via a prior distribution $\Pi\in\mathcal{P}(\Theta)$. Even if $\mathcal{M}^{(n)}$ is very complex, we assume that it is still feasible to generate
		synthetic observations $\bz$ according to $P^{(n)}_\theta$, for any $\theta\in\Theta$. Thus, even if the likelihood associated with $P^{(n)}_\theta$ is infeasible to calculate, useful information about the model can still be obtained by comparing observed data, $\by$, against simulated data, $\bz$. LFI methods can be used to conduct inference on $\theta$ by assigning posterior mass to values of $\theta$ that produce simulated data $\bz$ which is ``close-enough'' to $\by$. To make the problem practical from a computational perspective, LFI often resorts to matching low-dimensional summary statistics, defined by the map $S:\mathcal{Y}^n\rightarrow\mathcal{S}\subseteq\mathbb{R}^{d_s}$, and where we require that $d_s\ge d_\theta$. In what follows, when no confusion will result, we let $S$ denote the summary statistic mapping or the mapping evaluated at the observed data $\by$. 
		
		Given statistics $S$, the goal of LFI is to construct an approximation to the partial posterior $\pi(\theta|S)$. The two most well-established statistical approaches for constructing this posterior approximation are approximate Bayesian computation (ABC), see \citet{sisson2018handbook} for a review, and Bayesian synthetic likelihood (BSL), see \citet{wood2010statistical}, and \citet{price2018bayesian}. ABC and BSL differ in terms of how the posterior is approximated. In the case of ABC, the posterior is approximated by nonparametrically estimating the likelihood within the algorithm. In BSL, we approximate the intractable likelihood of the summaries using a normal density with mean $b(\theta)$ and variance $\Sigma(\theta)$. Since $b(\theta)$ and $\Sigma(\theta)$ are generally unknown, these are subsequently estimated via Monte Carlo using data simulated iid from $P^{(n)}_\theta$. 
		In what follows, we let $\widetilde\pi(\theta|S)$ denote an arbitrary approximation to the ``exact'' partial posterior $\pi(\theta|S)$.

		\subsection{Choosing Summaries}\label{sec:formal}

		Accurately approximating $\pi(\theta|S)$ becomes more computationally costly as the dimension of $S$, $d_s$, increases. Thus, the problem is to find a collection of summaries that are both low-dimensional and highly-informative about $\theta$. Many methods have been proposed to select summary statistics; we refer to \cite{blum2013}, and \cite{prangle2015summary} for in-depth reviews on different strategies. Several approaches are based on searching for informative subsets of summaries using information criteria such as AIC/BIC (\citealp{blum2013}), or entropy (\citealp{nunes2010optimal}), while other approaches are based on approximate sufficiency arguments (\citealp{joyce2008approximately,Chen2021a}).  In general, while such approaches can be useful, they lack a rigorous theoretical basis. 
		
		Alternatively, projection approaches seek to project an initial
		high-dimensional $S$ into a lower dimension space, and such methods have obtained much popularity in ABC applications. Arguably, the most celebrated of the projection approaches to summary statistic selection is the semi-automatic approach of \cite{FP2012}. \cite{FP2012} consider the problem of choosing summaries by attempting to give a decision rule $\delta\in\Theta$ that minimises the posterior expected loss
		\begin{flalign*}
			R_S(\delta)=\int (\theta-\delta)^\top  (\theta-\delta)
			\pi(\theta|\by)\dt\theta.
		\end{flalign*}\cite{FP2012} argue that $S=\E[\theta|\by]$ is the
		optimal choice of summary statistic, and that the minimum achievable loss based on the ABC posterior is achieved
		by the ABC posterior mean.  They propose to estimate $\E(\theta|\by)$ using (non)linear regression methods starting from an initial set of summaries. However, the goal of \cite{FP2012} is not to choose between summaries, but to approximate the most informative projection of a fixed initial set of summaries.  Hence posterior expected loss does not necessarily deliver a helpful criterion for deciding amongst competing collections of summaries. {Additional discussion regarding the difficulties
			involved in using $R_S(\delta)$ as a mechanism for choosing $S$ is given in Appendix A.1.}

		\subsection{Combining Information: Pooled Posteriors}\label{sec:pooled}
		While it is possible to choose a single vector of summaries to conduct inference on $\theta$,  we instead suggest to combine posterior inferences based on distinct sets of low-dimensional summary statistics. Such an approach obviates the need to conduct LFI using a high-dimensional vector of summaries, and still allows us to incorporate information contained across different sets of summaries. To make the following discussion as easily interpretable as possible, we restrict our attention here and in the sequel to the case where $S=(S_1^\top,S_2^\top)^\top$. It would be possible to extend our results to the general case of pooling $k$ approximate posteriors. However, since in general the optimal weights in such settings do not have a closed form (\cite{stone1961opinion}), we leave this extension for future research.

		Rather than choosing a single set, or attempting to conduct inference on $\theta$ using $S=(S_1^\top,S_2^\top)^\top$, we suggest to pool the inferences obtained from  $\widetilde\pi(\theta|S_1)$ and $\widetilde\pi(\theta|S_2)$ using a linear opinion pool (\citealp{stone1961opinion}): 
		\begin{align}
			\widetilde\pi_\omega(\theta|S) & :=(1-\omega)\widetilde\pi(\theta|S_1)+\omega\widetilde\pi(\theta|S_2), \label{lop}
		\end{align}
		where $\omega\in[0,1]$ controls the amount of mass assigned to each posterior. In particular, for a fixed pooling weight, $\omega$, the above posteriors can be sampled by generating posterior draws from $\widetilde\pi(\theta|S_1)$ and $\widetilde\pi(\theta|S_2)$, and mixing the draws with probability $\omega$. Such an approach to LFI is particularly useful in cases where $S_1$ and $S_2$ are relatively low-dimensional.  For a fixed computational budget, obtaining samples from $\widetilde\pi(\theta|S_1)$ and $\widetilde\pi(\theta|S_2)$ separately, which can be done in parallel, will be simpler than attempting to approximate the posterior $\widetilde\pi(\theta|S_1,S_2)$.
		
		To the best of our knowledge, the only other approach that considers a pooled posterior approach in the context of LFI is the work of \cite{chakra2022modular}, in which the authors are concerned with the application of LFI methods in the case of modular inference, and construct a linear  pool \textit{over a subset of posterior elements}. \cite{chakra2022modular} propose to select the pooling weight through prior-to-posterior conflict checks (see, e.g., \citealp{nott2020checking} for a discussion of such methods). In contrast, we consider an approach that is optimal for point estimation in terms of frequentist asymptotic risk, under appropriate conditions.  Also related to our work is model stacking, a technique for combining
		predictions from ensembles of models to improve posterior prediction under misspecification \citep{yao+vsg18}; see \cite{yao2023simulation} for a recent application of stacking in LFI.  However, our focus here is on 
		improving LFI posterior inference by combining posterior densities for
		different LFI methods or summaries, 
		rather than improved predictive
		inference via combining predictions from different models.
		
		A feature of the linear opinion pool \eqref{lop} is that variances for the parameters
		in the pooled posterior can be much larger than in any of the individual
		posteriors, particularly when the posterior means 
		$\mu_1:=\E(\theta|S_1)$ and $\mu_2:=\E(\theta|S_2)$
		are very different.  We can also consider the following modified opinion
		pool as an alternative to \eqref{lop}.  Write the mean of $\widetilde{\pi}(\theta|S)$
		as $\bar{\theta}(\omega):=(1-\omega)\mu_1+\omega \mu_2$, 
		and write the posterior covariance matrices of $\widetilde{\pi}(\theta|S_1)$
		and $\widetilde{\pi}(\theta|S_2)$ as $\text{Var}(\theta|S_1)$ and
		$\text{Var}(\theta|S_2)$ respectively.  
		We consider linear opinion pooling after recentering the components 
		$\widetilde{\pi}(\theta|S_1)$ and $\widetilde{\pi}(\theta|S_2)$ to 
		$\bar{\theta}(\omega)$.  We assume that the parametrization of the model
		is such that $\theta$ is unrestricted, which can be achieved by a transformation
		if necessary.  Defining recentered summary statistic posteriors by shifting
		location to $\bar{\theta}(\omega)$ by 
		\begin{align}
			\widetilde{\pi}_c(\theta|S_1):=\widetilde{\pi}(\theta-\mu_1+\bar{\theta}(\omega)|S_1),\;\;\;\;\;\widetilde{\pi}_c(\theta|S_2):=\widetilde{\pi}(\theta-\mu_2+\bar{\theta}(\omega)|S_2), \label{lopc}
		\end{align}
		and then our modified linear opinion pool is
		$$\widetilde\pi_c(\theta|S)=(1-\omega)\widetilde{\pi}_c(\theta|S_1)+\omega\widetilde{\pi}_c(\theta|S_2).$$
		Simple calculations show that the posterior mean and covariance for
		$\pi_c(\theta|S)$ are 
		$$\E_c(\theta|S):=\bar{\theta}(\omega),\;\;\;\;\;\text{Var}_c(\theta|S):=
		(1-\omega)\text{Var}(\theta|S_1)+\omega\text{Var}(\theta|S_2).$$
		The posterior mean is equal to $\bar{\theta}(\omega)$
		for both \eqref{lop} and \eqref{lopc}, 
		and the theory of Section 3 is concerned with point estimation using
		$\bar{\theta}(\omega)$, so that
		the theoretical results developed there apply to both cases.  
		On the other hand, for uncertainty quantification, \eqref{lopc} is less
		conservative than \eqref{lop} in the following sense.   
		Writing $\text{Var}(\theta|S)$ for
		the posterior covariance for \eqref{lop}, we can easily show that
		$\text{Var}_c(\theta|S)\leq \text{Var}(\theta|S)$, where 
		$A\leq B$ for two positive definite matrices $A$ and $B$ means that
		$B-A$ is non-negative definite. 
		
		Later we discuss a strong notion
		of misspecification commonly considered in LFI called ``incompatibility" and
		if $S_1$ is compatible, and $S_2$ is incompatible, we consider a
		data driven choice of the mixing weight $\omega$ which has
		the property that weight $1$ is assigned to the compatible summary asymptotically.  
		In this case, both \eqref{lop} and \eqref{lopc}
		give a pooled posterior of $\widetilde{\pi}(\theta|S_1)$, which is what 
		we would wish in terms of uncertainty quantification in this case.  
		When both summaries
		are compatible, \eqref{lop} and \eqref{lopc} can give conservative uncertainty
		quantification, but \eqref{lopc} is less conservative than \eqref{lop}.   
		We now discuss how to choose $\omega$ in an optimal way
		so that the quality of pooled posterior mean point estimation is improved
		compared to that of the posterior mean using either summary individually.
		

		\section{Optimality of pooled posteriors}\label{sec:optpost}
		In this section, we define an optimal pooling weight in terms of asymptotic risk for point estimation, and describe how it can be estimated. 
		To make the results in this section easier to state and follow, we maintain the following simplifying notations. For $x\in\mathbb{R}^{d}$, $\| x\| $ denotes the Euclidean norm of $x$. 
		Throughout, $C$
		denotes a generic positive constant that can change with each use. 
		For real-valued sequences $\{a_{n}\}_{n\geq 1}$ and
		$\{b_{n}\}_{n\geq 1}$: for $X_{n}$ a random variable, $X_{n}=o_{p}(a_{n})$ if
		$\lim_{n\rightarrow \infty }\text{pr} (|X_{n}/a_{n}|\geq C)=0$ for any $C>0, $
		and $X_{n}=O_{p}(a_{n})$ if for any $C>0$ there exists a finite $M>0$ and a
		finite $n'$ such that, for all $n>n'$, $\text{pr}(|X_{n}/a_{n}|\geq M)\leq C$. All limits are taken as $n\rightarrow\infty$, so that, when there is no confusion, $\lim_{n}$ denotes $\lim_{n\rightarrow\infty}$. The notation $\Rightarrow$ denotes weak convergence.  Let $\text{Int}(\Theta)$ denote the interior of the set $\Theta$. For any matrix $M\in\mathbb{R}^{d\times d}$, we define $|M|$ as the determinant of $M$, 
		and, let $\lambda_{\mathrm{max}}(M)$ and $\lambda_{\mathrm{min}}(M)$ be the maximal and minimal eigenvalues, respectively. For $f:\mathbb{R}^d\rightarrow\mathbb{R}$ a differentiable function of $x\in\mathbb{R}^d$, we take $\nabla_x f(x)$ to be the gradient and $\nabla^2_{xx}f(x)$ the Hessian. For a distribution $F$, we let $\E_F[X]$ denote the expectation of $X$ under $F$. When confusion is unlikely, we use $\E[X]$ to denote the expectation under the true distribution $P^{(n)}_0$. We use the notation $[M_1,M_2;M_3,M_4]$, for matrices $M_j, j=1,2,3,4$, with conformable dimensions, to denote the block partitioned matrix 
		$$\left[\begin{tabular}{cc}
			$M_1$ & $M_2$ \\
			$M_3$ & $M_4$
		\end{tabular}\right].$$

		The supplementary material
		contains proofs of all stated results. 
		
		\subsection{Defining optimal weights: asymptotic framework}\label{sec:asymp}
		
		{	To define an optimal pooling weight, let us first follow  \cite{FP2012} and consider the problem of choosing summaries by attempting to give a decision rule $\delta\in\Theta$ that minimises the posterior expected loss
			$\int L(\theta,\delta)
			\widetilde\pi(\theta|S)\dt\theta$, where  $L:\Theta\times\Theta\mapsto\mathbb{R}_{+}$ is a user-chosen loss function of interest. Under quadratic loss, $L(\theta,\theta')=\|\theta-\theta'\|^2$, \citet{FP2012} show that the posterior mean $\bar\theta=\int\theta\widetilde\pi(\theta|S)\dt\theta$ yields the smallest posterior expected loss, and, under regularity conditions, this result extends 
			asymptotically to any loss $L(\cdot,\cdot)$ satisfying certain assumptions; see Assumption~\ref{ass:loss} for specific details. However, as discussed in Section \ref{sec:formal}, and elaborated on in Appendix~A.1, posterior expected loss is not a helpful criterion for deciding amongst competing collections of summaries.

			Herein, we maintain the spirit of the minimum loss suggested in \citet{FP2012}, but instead define an optimal pooling weight by minimizing the asymptotic expected loss of the posterior mean for the pooled posterior $\bar\theta(\omega):=\int_\Theta \theta \widetilde\pi_\omega(\theta|S)\dt\theta$; see Section 5.5 of \cite{lehmann2006theory} for a discussion on asymptotic expected loss. Before we can formally define the optimal pooling weight obtained by minimizing this expected loss, we must first understand the asymptotic behavior of the pooled posterior mean $\bar\theta(\omega)$. 
			
			Recalling that $P^{(n)}_0$ denotes the true distribution of $\y$, we let $G^{(n)}_{j}$ denote the true distribution of $S_j(\y)$, the projection of $P^{(n)}_0$ under $S_j:\mathcal{Y}^n\rightarrow\mathcal{S}_j$. Denote the projection of the assumed model $P^{(n)}_\theta$, under $S_j$ as $F_{j,n}(\cdot|\theta)$. To characterize the optimal pooling weight, we consider two distinct situations: the first is where both sets of simulated summaries can match the observed summaries,  which has been termed \textit{compatibility} by \cite{Marinetal2014}, and the second is the case where only the first set of summaries is compatible. We treat the incompatible case in Section \ref{sec:incomp}, and focus here on the compatible case. }
		
		Formally defining compatibility requires some definitions and regularity conditions, which are similar to those encountered elsewhere in the literature on LFI;  
		see, in particular, \cite{Marinetal2014} and \cite{FMRR2016}. 	In the following assumptions, all matrices and vectors are partitioned
		conformably with $S(\y)=(S_1(\y)^\top,S_2(\y)^\top)^\top$.  
		
		\begin{assumption}\label{ass:sums} There exists a vector $b_{0}:=(b_{01}^\top,b_{02}^\top)^\top$ such that $\|S(\Y)-b_{0}\|=o_p(1)$.   There exists a sequence $\nu_n$ diverging to $+\infty$ such that 
			$
			\nu_n\{S(\Y)-b_{0}\}\Rightarrow N(0,V_{})$, under $P^{(n)}_0$, for some matrix $V=[V_1,\Omega_{1,2};\Omega_{1,2}^\top,V_{2}]$.
		\end{assumption}
		\begin{assumption}\label{ass:mapping}
			Let $b_{j}(\theta)$ denote the mean of $S_j(\Z)$ under $F_{j,n}(\cdot|\theta)$, with $b(\theta)=(b_1(\theta)^\top,b_2(\theta)^\top)^\top$. The following are satisfied for each $j$: (i) The mapping $\theta\mapsto b_j(\theta)$ is continuous and injective; (ii) For some matrix function $\theta\mapsto V(\theta)$, continuous and positive-definite for all $\theta\in\Theta$, $\nu_n\{S(\Z)-b(\theta)\}\Rightarrow N\{0,V(\theta)\},$ under $P_\theta^{(n)}$. 
		\end{assumption}
		
		{ A high-level interpretation of Assumptions \ref{ass:sums}-\ref{ass:mapping} are that they enable the summary statistics to produce asymptotically regular, i.e., asymptotically normal, inference. For an in-depth discussion of these assumptions see Remarks 1 and 3 in \cite{FMRR2016}. The following definition of compatibility between the assumed model and summary statistics formalizes when observed summaries can be matched \citep{Marinetal2014}. }
		
		\begin{definition}[Compatibility]
			The model $P^{(n)}_\theta$ and summaries $S$ are \textit{compatiass:mappingble} if there exist a unique $\theta_0\in\mathrm{Int}(\Theta)$ such that $b_{}(\theta)=b_{0}\iff\theta=\theta_0$.
		\end{definition}	
		
		{ Compatibility ensures that asymptotically the simulated summaries can match the observed values at a unique ``true value'' $\theta_0$. }	
		Under the above assumptions, and additional regularity conditions, it is possible to show that the posteriors $\widetilde\pi(\theta|S_1)$ and $\widetilde\pi(\theta|S_2)$ are asymptotically Gaussian. In the case of ABC, this result can be achieved under the assumptions of \cite{FMRR2016}, and for the case of BSL, see \cite{frazier2019bayesian}. 
		 Since these additional regularity conditions are not directly relevant to the form of the optimal pooling weight, and are specific to the precise LFI method employed in the analysis, we eschew these 
		in favour of the following high-level regularity condition. To state the condition, let $B_j(\theta)=\nabla_\theta b_{j}(\theta)$, $B_j=B_j(\theta_0)$, $\Sigma_j=(B_j^\top V_j^{-1}B_j)^{-1}$, 
		and let $$\Omega_\Sigma=Q_1\Omega_{1,2}Q_2^\top,\;Q_{j}=\Sigma_jB_j^\top V_{j}^{-1}.$$ Likewise, define the local parameter $t_j=\sqrt{n}(\theta-\theta_0)-Q_j\sqrt{n}\{S_j(\by)-b_j(\theta_0)\}$, and let $\widetilde{\pi}(t_j|S_j):=\widetilde{\pi}(\theta_0+t_j/\sqrt{n}+Q_j\sqrt{n}\{S_j(\by)-b_j(\theta_0)\}|S_j)$.	\begin{assumption}[Limiting Posteriors]\label{ass:posts}For $\widetilde\pi(t_j|S_j)$ the posterior for $t_j$, $
			\int \|t_j\||\widetilde\pi(t_j|S_j)- N(t;0,\Sigma^{}_j)|\dt t=o_{p}(1).
			$
		\end{assumption}		
		
		{Assumption \ref{ass:posts} maintains that the LFI posterior satisfies a Bernstein-von Mises result, i.e., that the posterior is asymptotically Gaussian. For ABC-based inference, Assumption \ref{ass:posts} is satisfied under the primitive regularity conditions outlined in \cite{FMRR2016}. Under the compatibility condition, the validity of Assumption \ref{ass:posts} can be ascertained by analysing the regularity of the simulated summaries, and, in particular, ensuring that they have appropriate moments so that they concentrate in a Gaussian manner over the support of $\Theta$.}
		
		{We maintain the following assumption on the user-chosen loss function, $\ell(\cdot)$, used in the analysis; this assumption requires that the loss is smooth in a neighbourhood of $\theta_0$, and assumes that the summaries are compatible. }
		\begin{assumption}\label{ass:loss}
			For any $\theta,\theta'\in\Theta$, $L(\theta,\theta')=\ell(\|\theta-\theta'\|)$, for some known function $\ell(\cdot)$ such that $\ell(0)=0$, there exists a $\delta>0$, such that for all $\theta\in\Theta$ with $\|\theta-\theta_0\|\le\delta$, $\ell(\|\theta_0-\theta\|)$ is three times continuously differentiable in $\theta$ with: (i) $\nabla_\theta\ell(\|\theta-\theta_0\|)|_{\theta=\theta_0}=0$; (ii) For $H(\theta)=\nabla_{\theta\theta}^2L(\theta_0,\theta)$, $H_0:=H(\theta_0)$ is positive-definite.
		\end{assumption}

		{The regularity conditions in Assumptions \ref{ass:sums}-\ref{ass:loss} allow us to define the optimal pooling weight $\omega$ as the value that minimizes the trimmed asymptotic loss of the pooled posterior:
			\begin{flalign*}
				\mathcal{R}_{0}(\omega)&:=\lim_{\nu\rightarrow\infty}\liminf_{n\rightarrow\infty}\E\left[\min\{nL\{\theta_0,\overline\theta(\omega)\},\nu\}\right];
			\end{flalign*}the asymptotic expected loss $\E\left[nL\{\theta_0,\overline\theta(\omega)\}\right]$ is trimmed at $\nu$ so that $
			\mathcal{R}_{0}(\omega)$ is guaranteed to exist. The following result shows that the optimal pooling weight has a simple form when both sets of summaries are compatible.   
		}

		\begin{lemma}\label{lem:gamstar}
			Under Assumptions \ref{ass:sums}-\ref{ass:loss},  $\mathcal{R}_{0}(\omega)$ is minimised at $\omega^\star_+:=\min\{1,\omega^\star\}$, where 
			\begin{equation}\label{eq:gammastar}
				\omega^\star=\begin{cases}\frac{\tr H_0(\Sigma_1-\Omega_{\Sigma})}{\tr H_0\Sigma_1+\tr H_0\Sigma_2-2\tr H_0\Omega_{\Sigma}}& \text{ if }\tr H_0\Sigma_1>\tr H_0\Omega_{\Sigma}\\
					0& \mathrm{ otherwise }
				\end{cases}.
			\end{equation}	
		\end{lemma}

        We can give some intuition on $\omega^\star$ by considering the special case, when both $S_1$ and $S_2$ are univariate. In this case, $V=[\sigma_1^2,\rho\sigma_1,\sigma_2;\rho\sigma_1\sigma_2,\sigma_2^2]$, where $\sigma_j^2$ is the (asymptotic) variance of $S_j(\Y)$ under $P_0^{(n)}$ and $\rho$ denotes the correlation between $S_1$ and $S_2$. Then, $\tr H_0\Sigma_1>\tr H_0\Omega_\Sigma$ if and only if $\sigma_1-\rho\sigma_2B_1^\top B_2\left(B_2^\top B_2\right)^{-1}>0$, and in this case 
        \begin{align*}
            \omega^\star=\frac{\sigma_1^2B_2^\top B_2-\rho\sigma_1\sigma_2B_1^\top B_2}{\sigma_1^2B_2^\top B_2-2\rho\sigma_1\sigma_2B_1^\top B_2+\sigma_2^2B_1^\top B_1}.
        \end{align*}
        Thus, if the correlation $\rho$ is small, the weight assigned to $\widetilde{\pi}(\theta\mid S_j)$ by $\omega^\star$ is approximately proportional to $\sigma_j^{-2}B_j^\top B_j$, $j=1,2$ and thus to the precision of the asymptotic posterior. On the other hand, if $\vert\rho\vert$ is large, $\omega^\star$ is potentially dominated by $\sigma_1\sigma_2B_1^\top B_2$. This indicates that for similar $S_1$ and $S_2$, $\omega^\star$ is close to $0.5$.

		\subsection{Alternative pooling weights}
		A pooled posterior based on $\omega^\star_+$ will asymptotically have an expected loss that is weakly smaller than either individual posterior.  
		This means that defining $\mathcal{R}_0(S_1):=\mathcal{R}_0(0)$ and $\mathcal{R}_0(S_2):=\mathcal{R}_0(1)$, 
		$
		\mathcal{R}_0(\omega^\star_+)\le \min\{\mathcal{R}_0(S_1),\mathcal{R}_0(S_2)\}
		$.   
		In practice, an estimator of $\omega^\star$ can be obtained using information from both posteriors, and any consistent estimator for the covariance term $\Omega_\Sigma$. In Appendix~A.2 we give specific details as to how $\Omega_\Sigma$ can be estimated. 
		Given an estimator $\overline\Omega_{\Sigma}$ of $\Omega_\Sigma$, and estimators $\bar\theta_1=m^{-1}\sum_{i=1}^{m}\theta_{j,i}$,  $\overline\Sigma_1=m^{-1}\sum_{i=1}^{m}(\theta_{j,i}-\bar\theta_j)(\theta_{j,i}-\bar\theta_j)^\top$, $\theta_{j,i}\stackrel{iid}{\sim}\widetilde\pi(\theta|S_j)$, we can estimate $\omega^\star$ using $\widehat\omega^\star_+=\min\{1,\widehat\omega^\star\}$, where, for $\overline{H}=H(\bar\theta_1)$,  
		$$
		\widehat\omega^\star:=\begin{cases}
			\frac{\tr \overline{H}(\overline\Sigma_1-\overline{\Omega}_\Sigma)}{\tr \overline{H}\overline\Sigma_1+\tr \overline{H}\overline\Sigma_2-2\tr \overline{H}\overline{\Omega}_\Sigma}&\tr \overline{H}\overline\Sigma_1>\tr \overline{H}\overline{\Omega}_\Sigma\\0&\text{ otherwise }
		\end{cases}.
		$$
		
		In finite samples, estimation of $\Omega_\Sigma$ can inject additional noise into the pooled posterior, which may lead to a degradation in the accuracy of the pooling approach. {Poor empirical performance for pooling weights based on plug-in estimators is so ubiquitous in the literature on combination methods that this phenomenon is called the combination puzzle; see \cite{wang2022forecast} for a review.   
			
			Given the difficulties associated with estimation of $\Omega_\Sigma$, and the ensuing ill-effects, we propose two alternatives that do not require estimation of $\Omega_\Sigma$. The first approach is precisely the weight $\widehat\omega^\star$ but where we artificially set $\overline\Omega_\Sigma=0$, to obtain
			$$
			\widehat\omega_{}:=\frac{\tr \overline{H}\overline\Sigma_1}{\tr \overline{H}\overline\Sigma_1+\tr\overline{H}\overline\Sigma_2}.
			$$
            Setting $\overline\Omega_\Sigma=0$ is well motivated in many common cases. For example, if $\Omega_{1,2}$ is close to $0$, which indicates low correlation between $S_1$ and $S_2$, so is $\Omega_\Sigma$.
            
            However, the pooling weight $\widehat\omega$ disregards the fact that the posteriors $\widetilde\pi(\theta|S_1)$ and $\widetilde\pi(\theta|S_2)$ can  have distinct locations. To account for this fact, while incorporating the structure of $\widehat\omega$, we also propose the alternative pooling weight 
			$$
			\widetilde\omega_{}:=\frac{\tr\overline{H}\overline\Sigma_1}{(\bar\theta_1-\bar\theta_2)^\top (\bar\theta_1-\bar\theta_2)+\tr\overline{H}\overline\Sigma_1+\tr\overline{H}\overline\Sigma_2}
			. $$
			{The weight $\widetilde\omega_{}$ is particularly useful when the summary
				statistics $S_1$ are thought to provide reliable inferences, and where we are unsure of the models ability to match the summary statistics
				$S_2$.  Hence, if the posterior location is very different
				for the component LFI posteriors, a higher weight is assigned 
				to the $S_1$ component.  This means that the two LFI component posteriors are not
				treated symmetrically under this pooling weight.}
			
			Critically, these alternative pooling weights can be estimated using only samples from the constituent posteriors; no estimation of $\Omega_\Sigma$ is required. Obtaining the pooled posterior, based on $\widehat\omega$ or $\widetilde\omega$, is as simple as sampling from $\widetilde\pi(\theta|S_1)$ and $\widetilde\pi(\theta|S_2)$. Furthermore, in the case where the summaries are compatible, the two weights, $\widehat\omega$ and $\widetilde\omega$, will agree asymptotically: that is, under our assumptions, $$\widehat\omega=\widetilde\omega+o_p(1)=\omega_0+o_p(1),\quad \omega_0:=\frac{\tr H_0\Sigma_1}{\tr H_0\Sigma_1+\tr H_0\Sigma_2}.$$

			While the optimal pooling weight depends on the covariance term $\Omega_\Sigma$, it is not difficult to see that if $\tr H_0\Omega_\Sigma$ is small, then the simpler pooling weights will be close to the optimal weight. More generally, the simpler weights will always perform better than using $S_1$ or $S_2$ alone, in terms of risk, in the following empirically relevant scenarios. 
			
			\begin{lemma}\label{lem:naive1}
				If $\tr H_0\Omega_\Sigma\le \frac{1}{2}\min\{\mathcal{R}_0(S_1),\mathcal{R}_0(S_2)\}$, then
				$
				\mathcal{R}_0(\widehat\omega)\le \min\{\mathcal{R}_0(S_1),\mathcal{R}_0(S_2)\}.
				$
				
			\end{lemma}
			{Lemma \ref{lem:naive1} demonstrates that if the trace of the covariance $\tr H_0\Omega_\Sigma$ is negative, or small, then the pooled posterior will perform better than using the posterior for $S_1$ or $S_2$ individually. The above condition can be checked in cases where the posterior covariance can be estimated reliably. Consider again the special case that $S_1$ and $S_2$ are univariate. Then, $\tr H_0\Omega_\Sigma=\rho\sigma_1\sigma_2 B_1^\top B_2\left(B_1^\top B_1 B_2^\top B_2\right)^{-1}\tr H_0$ and thus, the condition of Lemma~\ref{lem:naive1} is for example fulfilled if the covariance between $S_1$ and $S_2$ is low.  However, it is not guaranteed to be satisfied in all settings. In Appendix~A.3, we give an example where the pooled posterior is outperformed by a particularly informative collection of summaries, which produces a small posterior variance, and has posterior means that are also well-located. As a consequence, the pooled posterior does not {produce more accurate inferences than those based solely on the more informative collection. However, the differences between the pooled results and best performing results, as measured by MSE, are relatively small, and the pooled posterior still produces accurate inferences.}} 

Under certain conditions it is possible to derive  $\mathcal{R}_0(S_1)$ and $\mathcal{R}_0(S_2)$ explicitly, and thus give an analytical representation for when the pooled posteriors will outperform either collection. 
\begin{lemma}\label{lem:bias}
Under Assumption \ref{ass:sums}-\ref{ass:posts},  for $V_\Sigma=[\Sigma_1,\Omega_{\Sigma}^\top;\Omega_\Sigma,\Sigma_2]$,
$(\sqrt{n}(\bar\theta_1-\theta_0)^\top,\sqrt{n}(\bar\theta_2-\theta_0)^\top)^\top\Rightarrow (\xi+\tau)$ where $\xi=(\xi_{1}^\top,
\xi_{2}^\top)^\top\sim N\left(0,V_\Sigma\right)	
$ and $\tau=(0^\top,[Q_2\tau_2]^\top)^\top$.
\end{lemma}

\begin{lemma}\label{corr:amse}Consider that Assumptions~\ref{ass:posts}-\ref{ass:loss} are satisfied. If $\|\tau_2\|<\infty$,  $\mathcal{R}_0(S_1)=\tr\left\{H_0\Sigma_1\right\}$ and $\mathcal{R}_{0}(S_2)=\tr\left\{H_0\Sigma_2\right\}+\tau_2^\top Q_2^\top  H_0Q_2\tau_{2}$. 
\end{lemma}
}
}
			
\subsection{Examples}
It can be challenging to assess whether or not Assumptions~\ref{ass:sums}-\ref{ass:loss} hold. However, in practice we find 
the derived weights widely useful, even in settings where verification
of the assumptions is infeasible. We now compare three suggested choices for the pooled posterior in two toy examples commonly used in the LFI literature. Squared error loss is used, so that asymptotic risk is equivalent to asymptotic mean squared error. This loss fulfills Assumption~\ref{ass:loss}. Across the weight choices $\widehat\omega_+^\star$, $\widehat\omega$ and $\widetilde\omega$, and across all experiments, we find that the worst performing pooled posterior is based on $\widehat\omega_+^\star$. We conjecture that the poor performance is due to the additional noise introduced when estimating the covariance matrix $\Omega_\Sigma$.
			\subsubsection{Example: g-and-k}\label{sec:gandk}
			The g-and-k model has an intractable likelihood and is often used as a test case in the LFI literature (see, e.g., \citealp{FP2012}).   The model is defined through its quantile function: 
			\begin{align}
				Q\{z(p);{\theta}\} &= a + b\left[ 1+c\frac{1-\exp\{-gz(p)\}}{1 + \exp\{-gz(p)\}}\right]\{1 + z(p)^2\}^kz(p) \label{eq:g-and-k},
			\end{align}
			where $p\in(0,1)$, $z(p)$ is the quantile function of the standard normal distribution, and the model parameters are ${\theta} = (a,b,g,k)^\top $, while the parameter $c$ is fixed at 0.8 (see \citealp{rayner2002numerical} for discussion). Similar to \citealp{FP2012} we use a uniform prior on $[0,10]^4$.

			We compare the pooled posterior approach based on two different sets of summaries.  The first set of summaries $S_1$ has dimension 4 and 
			was proposed by \cite{drovandi2011likelihood}. For $S_1$ we use the summaries proposed in \cite{drovandi2011likelihood}, so that $S_1=(S_{11},S_{12},S_{13},S_{14})^\top$, where
				$S_{11}=L_{2}$, $S_{12}=L_{3}-L_{1}$, $S_{13}=S_{12}^{-1}(L_{3}+L_{1}-2 L_{2})$, $S_{14}=S_{12}^{-1}(E_{7}-E_{5}+E_{3}-E_{1})$, and where $L_i$ denotes the $i$-th quartile and $E_i$ the $i$-th octile. The components of $S_1$ are robust estimates of location, scale, skewness and kurtosis.
			The second set of summaries $S_2$ has dimension 7 and consists of the seven sample octiles.  We also compare the pooled posteriors against the posterior that uses summaries $S=(S_1^\top,S_2^\top)^\top$; the latter is more expensive to sample from, but allows us to quantify the information that is lost in the posterior pooling. We sample the posteriors $\widetilde\pi(\theta|S_1)$, $\widetilde\pi(\theta|S_2)$ and $\widetilde\pi(\theta|S)$ using the ABC-SMC algorithm of \cite{drovandi2011estimation}, where we stop the algorithm when the acceptance rate drops below 5\% and generate 1000 sample draws from each posterior. 
			
			The first pooled posterior we compare is based on $\widehat\omega^\star_+$, where the variance and covariance matrices in $\Omega_\Sigma$ are estimated using a standard iid bootstrap. Precise details are given in Section~A.2 of the supplementary material. Two additional pooled posteriors are considered using the estimated weights $\widehat\omega$ and $\widetilde\omega$, respectively. We simulate 100 synthetic samples of size $n=1000$ from the g-and-k model under true parameter value $\theta_0=(a_0,b_0,g_0,k_0)^\top=(3,1,2,0.5)^\top$. Across each method, the following averages across the replications for each parameter are reported in Table~\ref{tab:robbase_oct}: the bias of the posterior mean, the posterior standard deviation and the raw MSE of the marginal posterior mean.  The overall MSE, i.e., the sum of raw MSE across the different parameters, is also reported in the table caption. 
			
			{	The pooled posterior approach based on the naive pooling choices  $\widehat\omega$ and $\widetilde\omega$ produces inferences that are more accurate - in terms of bias and variance - than using either individual posterior. In comparison with the posterior based on all the summaries, the ranking is less clear, with the pooled posterior producing smaller biases and standard deviations than the joint posterior for some parameters, and the reverse being true for others.
				
				The best performing pooled posterior according to total MSE is $\widetilde\omega$. This posterior obtains a 65\% reduction in MSE across the experiments relative to $S_1$ alone, while a much smaller 5\% reduction is achieved relative to the posterior for $S_2$. Notably, the posterior $\widetilde\pi(\theta\mid S)$ has an MSE that is only about 5\% smaller than that of the pooled posteriors based on $\widetilde\omega$ and $\widehat\omega$. $\Bar{\theta}_2$ has smaller MSE then $\Bar{\theta}_1$. Consequently, the average weight on the second set of summaries under $\widehat\omega$ is  $0.8711$. The weight based on the estimated covariance is much closer to 1/2 on average and leads to a much less accurate pooled posterior.
				
				As hypothesised, the additional sampling variability that is required to compute the optimal pooling weight $\widehat\omega^\star_+$ delivers inferences that are less accurate than the infeasible weight $\omega^\star_+$. Given these results, we believe the more naive weights $\widehat\omega$ and $\widetilde\omega$ are likely to produce more reliable inferences on average than pooled posteriors based on  the estimated optimal pooling weight $\widehat\omega^\star_+$.}

			\begin{table}[H]
				\centering
				{\footnotesize
					\begin{tabular}{lrrrrrrrrr}
						\hline\hline
						& \multicolumn{3}{c}{$S_1$} & \multicolumn{3}{c}{$S_2$} & \multicolumn{3}{c}{$S=(S_1,S_2)$} \\ 
						\textbf{} & Bias        & Std    & MSE       & Bias       & Std      &MSE        & Bias     & Std      &MSE \\ \hline
						$a$ &   -0.0312  &  0.0462 &   0.0045  &0.0002  &  0.0180  &  0.0017 &0.0003 &   0.0194  &  0.0017\\
						$b$ &   0.0089  &  0.0958 &   0.0169   &0.0125  &  0.0392  &  0.0086 & 0.0132  &  0.0338  &  \textbf{0.0081} \\
						$g$ &  0.2804 &   0.2777  &  0.1408     &0.0159 &  0.1041   &  0.0445 & 0.0188  &  0.1102  &  \textbf{0.0410} \\
						$k$  &    0.0530  &  0.1118 &   0.0234  &0.0167 &   0.0482 &   0.0143&-0.0145   &   0.0370 &   \textbf{0.0108} \\\hline\hline
						& \multicolumn{3}{c}{$\widehat\omega^{\star}_+$} & \multicolumn{3}{c}{$\widehat\omega$} & \multicolumn{3}{c}{$\widetilde\omega_{}$}   \\ 
						\textbf{} & Bias    & Std      & MSE      & Bias      & Std        &MSE       & Bias    & Std      &MSE \\ \hline
						$a$ & -0.0200 &   0.0310 &   0.0028 &   -0.0018&    0.0175 &   \textbf{0.0016} &  -0.0044 &   0.0184 &   \textbf{0.0016}\\
						$b$ &-0.0006 &   0.0650 &   0.0117 &   0.0112 &   0.0377 &   0.0084 &   0.0094 &   0.0394 &   0.0083\\
						$g$ &0.0434  &  0.1860  &  0.0793 &   0.0183 &   0.1038  &  0.0429 &   0.0213 &   0.1096  &  0.0422\\
						$k$ &-0.0060 &   0.0760 &   0.0167 &  -0.0144&    0.0455 &   0.0138 &  -0.0116  &  0.0471 &   0.0134\\\hline\hline
				\end{tabular}}
				\caption{Posterior accuracy results  in the g-and-k model under the base set of summaries $S_1$ (robust summaries), the alternative set $S_2$ (octiles), and the pooled posteriors ($S$). Bias is the bias of the posterior mean for $\theta_0$ across the replications. Std is the average posterior standard deviation across the replications, and MSE the mean squared error. For each parameter the smallest MSE across methods is given in bold. The overall MSE over the replications is:  $S_1$: 0.1856; $S_2$: 0.0690  \textbf{$S$: 0.0617};  $\widehat\omega^\star_+$: 0.1105;  $\widehat\omega$:    0.0668 $\widetilde\omega$:    0.0655.}
				\label{tab:robbase_oct}
			\end{table} 
			
			\subsubsection{Example: Stochastic Volatility Model}\label{sec:maexam}
			Consider a simple stochastic volatility model of order one, where observed data is generated according to 
			\begin{equation}
				y_{t}=\exp(h_t/2)e_{t},\quad h_t=\zeta+\rho h_{t-1}+\sigma_v\nu_t,  \quad t=1,\dots,n,\quad \label{eq:SV1}
			\end{equation} where $(e_t,\nu_t)^\top$ are iid standard normal, $h_0\sim N\left(\frac{\zeta}{(1-\rho)},\frac{\sigma^2_v}{(1-\rho^2)}\right)$, and the unknown parameters are $\theta=(\zeta,\rho,\sigma_v)^\top$. Our prior distribution for $\theta$ is uniform over $(-1,1)\times(0,1)\times(0,1)$.

			\cite{martin2019auxiliary} demonstrate that useful summary statistics for this model can be obtained by first taking squares and logarithms of the process to notice that 
			$$
			y_t^\ast=\log y_t^2=\log e_t^2 +\zeta +\rho h_{t-1}+\sigma_v\nu_t,
			$$
			which resembles a latent autoregressive process of order one. Consequently, we can use summary statistics, in $\log y_t^2$, that would identify the parameters of an observable autoregressive model.  For the auxiliary autoregressive model
			$$
			y_t^\ast=\beta^\top X_t+\epsilon_t,\quad X_t=[1,\log y^2_{t-1},\log y^2_{t-2}]^\top,\quad t=3,\dots,n,
			$$ we write $\widehat\beta$ for the estimated regression coefficient for the observed data.  The observed three dimensional summaries are given by $S_1(\y)=\sum_{t=3}^{T}X_t(y_t^\star-\widehat\beta^\top X_t)$.

			In addition to sample moments from an auxiliary model, unconditional sample moments for data from the stochastic volatility model are known to provide reliable point estimators of the unknown parameters (\citealp{andersen1996gmm}), and so matching sample moments of the data should also provide reliable summary statistics. We consider four sample 
			moments based on the absolute value of powers of the observed data, i.e., $|y_t^k|$, $k=1,2,3,4$, and the first three sample autocovariances, i.e., $y_ty_{t-k}$, $k=1,2,3$. The resulting seven-dimensional summary statistic is denoted
			$S_2$. 
			
			{Again, we compare the accuracy of the pooled posteriors against the individual and joint posteriors. We apply these approaches to 100 synthetic datasets of size $n=1000$ generated from \eqref{eq:SV1} under the true parameter value $\theta_0=(-0.74,0.90,0.36)^\top$. Similar to the previous experiment, we apply three different pooling approaches based on the estimated pooling weight. In this example, we use the default sampling options in the R package \texttt{bsl} (\citealp{JSSv101i11}) to produce posterior samples from $\widetilde\pi(\theta\mid S_1)$, $\widetilde\pi(\theta\mid S_2)$ and $\widetilde\pi(\theta\mid S_1,S_2)$. For each posterior we obtain $5000$ MCMC samples that are based on using 100 synthetic datasets to estimate the mean and variance of the summaries. }

			{The results are presented in Table \ref{tab:sv1auxbase}. Similar to the g-and-k example, the pooled posteriors are more accurate than either individual posterior. Relative to the sample moment summaries, $S_1$, the pooled posterior based on $\widehat\omega$ obtains a nearly 53\% reduction in MSE across the experiments, while a 45\% reduction was achievable relative to the posterior based on $S_2$ and the posterior based on $S$. Across most parameters, both the bias and standard deviation of the pooled posteriors are smaller than that achieved by the posterior based on $S$. Hence, for the fixed computational budget employed in this experiment, the pooled posteriors are much more accurate than the posterior based on $S$. This is at least partially due to the aforementioned curse of dimensionality that makes LFI with a high-dimensional vector of summaries challenging.
			}
			
			\begin{table}[H]
				\centering
				{\footnotesize
					\begin{tabular}{lrrrrrrrrr}
						\hline\hline
						& \multicolumn{3}{c}{$S_1$} & \multicolumn{3}{c}{$S_2$} & \multicolumn{3}{c}{$S=(S_1,S_2)$} \\ 
						\textbf{} & Bias     & Std       & MSE    & Bias     & Std      &MSE       & Bias     & Std      &MSE \\ \hline
						$\zeta$  & 0.0214 &   0.2133  &  0.0452 &  0.0232   &   0.2108 &   0.0466 &  0.0286 &   0.2127 &   0.0476\\
						$\rho     $& -0.0001&  0.0289  & 0.0011 &  0.0030    &   0.0321 &   0.0009&    0.0037 &   0.0290 &   0.0009\\
						$\sigma_v$&-0.0390  & 	0.0753  &  0.0177 &  -0.0169 &   0.1101  &  0.0073 &  -0.0254 &  0.0695 &  0.0064\\\hline\hline
						& \multicolumn{3}{c}{$\widehat\omega^{\star}_+$} & \multicolumn{3}{c}{$\widehat\omega$} & \multicolumn{3}{c}{$\widetilde\omega_{}$}   \\ 
						\textbf{} & Bias &       Std & MSE      & Bias       & Std      &MSE       & Bias      & Std     &MSE \\ \hline
						$\omega$ &0.0224   &   0.1746   &  0.0312 & 0.0223     &   0.1508 &  \textbf{0.0234} & 0.0224    &  0.1521 &   0.0238\\
						$\rho$ &0.0019      &  0.0244   &   0.0006  &  0.0016   &   0.0217  &  0.0006 & 0.0017    &  0.0216  &  \textbf{0.0005}\\
						$\sigma_v$ & -0.0254 &   0.0701 &  0.0068  &  -0.0267 &   0.0651 &   0.0068  &  -0.0260  &  0.0637  &\textbf{0.0059}\\\hline\hline
				\end{tabular}}
				\caption{Posterior accuracy results in the stochastic volatility model under the base set of summaries $S_1$ (sample moments), the alternative set $S_2$ (autoregressive summaries), and the pooled posterior ($\omega$). The remaining information is as in Table \ref{tab:robbase_oct}. The overall MSE over the replications is:  $S_1$: 0.0639; $S_2$: 0.0574;  $S$: 0.0549;  $\widehat\omega^\star_+$: 0.0386; $\widehat\omega$:    \textbf{0.0301}; $\widetilde\omega$: 0.0302.} 
				\label{tab:sv1auxbase}
			\end{table}

			\subsection{Incompatible summaries}\label{sec:incomp}
			We now study the case where only one set of summaries is compatible, while the other is incompatible. We assume that, either by prior knowledge or previous studies, there  is a subset $S_1$ of $S$ that we believe is compatible, with $S_1\in\mathcal{S}_1\subseteq\mathbb{R}^{d_{1}}$, and $d_1\ge d_\theta$.  The set $S_2$ is possibly incompatible: there exist $\theta_0\in\Theta$ such that $b_{1}(\theta)=b_{0,1}\iff\theta=\theta_0$; while $b_2(\theta_0)\ne b_{0,2}$. 
			
			In this case,  we can show that (in large samples) the pooled posterior approach based on $\widetilde\omega$, places zero weight on the second set of summaries if they are in fact incompatible. 
			\begin{corollary}\label{cor:simple}
				Assume that Assumptions \ref{ass:sums}-\ref{ass:loss} are satisfied for $\widetilde\pi(\theta|S_1)$. If there exists some $\theta^\star\ne \theta^0$ such that  $\sqrt{n}(\bar\theta_2-\theta^\star)=O_p(1)$, and $\overline\Sigma_2=\Sigma_2+o_p(1)$, $\|\Sigma_2\|>0$, then $\omega^\star=0$ and $\widetilde\omega=o_p(1).	$
			\end{corollary}  
			
			Corollary \ref{cor:simple} demonstrates that if the summaries $S_1$ are compatible, but $S_2$ are incompatible, then the pooling weight converges to zero in probability; i.e., in large samples the pooled posterior places weight $1$ on the compatible set $S_1$. Of course, such a result requires that $S_1$ is compatible. A reasonable empirical check for compatibility is to see if the observed summaries fall within the region of support for the posterior predictive distribution of the summaries.  Alternatively, one can use the methods suggested by \cite{Marinetal2014} and \cite{frazier2021robust} to check whether or not the summaries $S_1$ are compatible.

			{When the summaries are not compatible, the behavior of LFI posterior means has not been formally established in all cases and \cite{frazier2021synthetic} show that the posterior mean may not even be asymptotically normal. Consequently, the theoretical results obtained in Lemma \ref{lem:naive1} and Corollary \ref{cor:simple} will not be satisfied, and determining the behavior of the pooling weight becomes difficult. Consequently, if posterior predictive analysis suggests that all summaries are not compatible, we suggest to instead conduct robust LFI using the approaches suggested by \cite{frazier2021robust}, which has been generalized to more complex LFI settings by \cite{kelly2023misspecification}. 
			}	
			
			\subsection{Example: individual-based model of toad movement}
			
			Here we consider the individual-based movement model of Fowler's Toads (\textit{Anaxyrus fowleri}) of \citet{marchand2017stochastic}, which has also been used as an illustrative example in other likelihood-free research (e.g.\ \citet{drovandi2021comparison}).  Here we only provide minimal details of the example and refer to \citet{marchand2017stochastic} and \citet{drovandi2021comparison} for more  information.
			
			The model has three parameters, $\theta = (\alpha,\xi,p_0)^\top$.  The overnight displacement for each toad is drawn from a Levy alpha-stable distribution, parameterised by $\alpha$ and $\xi$.   \citet{marchand2017stochastic} consider three models for how each toad takes refuge during the day. Here we consider their `Model 2' since there is evidence that the model does not provide a good fit to the data.  In this model, each toad will take refuge at the closest refuge site it has previously visited with a probability $p_0$, otherwise it will take refuge at the new location.   The empirical data consist of GPS location data for 66 toads for 63 days.  In \citet{marchand2017stochastic} the data is summarised down to four sets comprising the relative moving distances for time lags of $1,2,4,8$ days.  For each lag, we record the number of returns and the distances for the non-returns.  We further summarise the vector of non-return distances by 11 equally spaced quantiles.  For each time lag, there are thus 12 summary statistics (including the number of returns).  
			
			Anticipating that the model can capture data related to a lag of 1 day, but does not provide a good fit for longer time lags, we run two separate ABC analyses, one which just includes lag 1 summaries and another that includes summaries for the remaining lags, thus $\mbox{dim}(S_1) = 12$ and $\mbox{dim}(S_2) = 36$.  In each case we use the ABC-SMC algorithm of \cite{drovandi2011estimation} to sample the approximate posterior.   We find that the observed summaries for lag 1 are compatible with the model, while some summaries for the remaining lags lie in the tail of the posterior distribution of the summaries.   The estimated univariate posteriors of the parameters are shown in Figure \ref{fig:posteriors_toad}.   There is some indication of a difference in the posteriors between the two ABC analyses.  From pooling the two ABC analyses, an estimated $\tilde{\omega} = 0.061$ is obtained, which suggests placing a large weight on the ABC results based on the compatible lag 1 summaries, consistent with the theoretical results above.

			\begin{figure*}
				\centering
				\includegraphics[width=\textwidth,keepaspectratio]{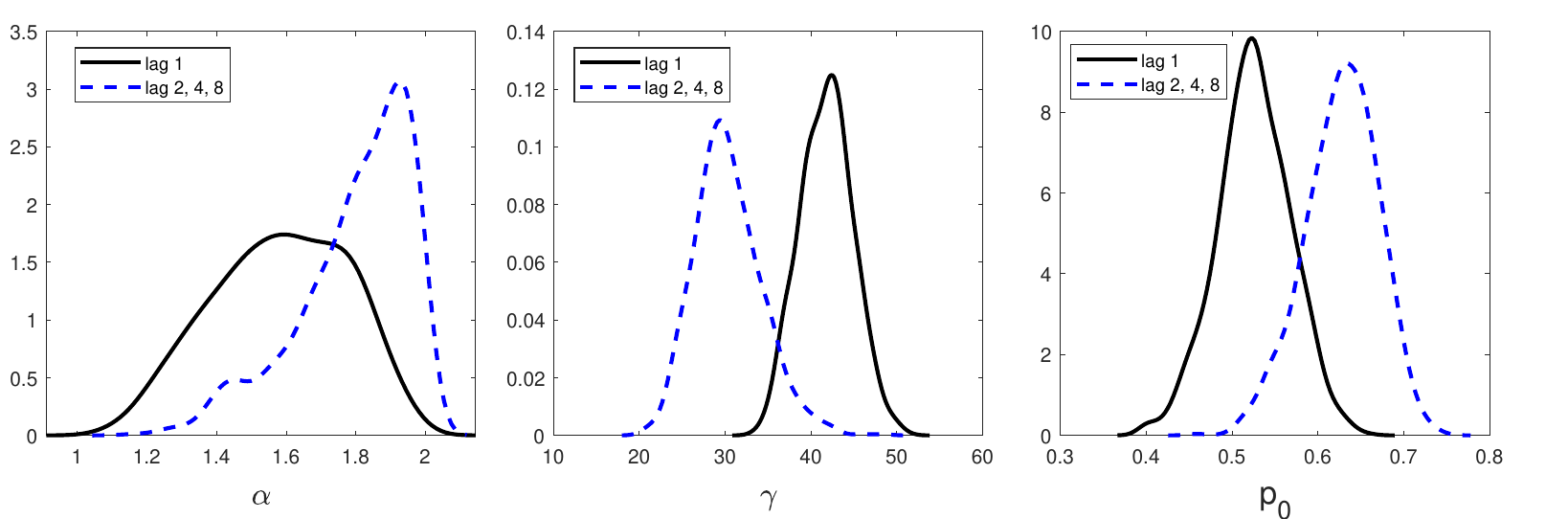}
				\caption{Estimated univariate posterior distributions for the parameters of the toad example.  Shown are the results for lag 1 summaries (solid) and the results for the remaining lags (dash).}
				\label{fig:posteriors_toad}
			\end{figure*}

			\section{Pooling different types of posteriors}\label{sec:poolBSL+ABC}
			Whilst the above analysis has so far focused on combing LFI posteriors built using different summary statistics, the pooled posterior approach is also applicable if we wish to combine summary statistic-based posteriors and posteriors built using general discrepancy measures between the observed and simulated data. Recently, several authors have suggested replacing the distance and summary statistics under which LFI is usually implemented with distances based on empirical measures. For a review of such methods, we refer to \cite{drovandi2021comparison}. The benefit of such methods are that they do not require a choice of summary statistics, however, as documented by \cite{drovandi2021comparison}, such methods may deliver inferences that are not as precise as those obtained under an informative set of summaries.

			Let $\mathcal{D}:\mathcal{Y}^n\times\mathcal{Y}^n\rightarrow\mathbb{R}_+$ denote a discrepancy function used to measure the difference between the observed data $\y$ and data $\z$ simulated under the model $P_\theta^{(n)}$. We assume that the observed and synthetic data have the same sample size. An ABC-based posterior for $\theta$ under $\mathcal{D}(\y,\z)$ can then be sampled using a number of different algorithms, such as accept/reject ABC or Markov chain Monte Carlo ABC (ABC-MCMC). In the experiments that follow we use a tuned version of the ABC-MCMC algorithm, see \cite{sisson2011likelihood} for a review, to obtain samples from the approximate posterior $\widetilde\pi(\theta|\mathcal{D})$.

			Given a posterior based on summaries $S$, $\widetilde\pi(\theta|S)$, and a posterior based on $\mathcal{D}$, $\widetilde\pi(\theta|\mathcal{D})$, we can pool the posteriors via
			$$
			\widetilde\pi_\omega(\theta|\y):=\omega\widetilde\pi(\theta|S)+(1-\omega)\widetilde\pi(\theta|\mathcal{D}). 
			$$
			Similar to \eqref{lopc}, we can also recentre the mixture components before
			pooling, which does not change the posterior mean for the pooled posterior. 
			Denote the estimated posterior variance obtained under $\widetilde\pi(\theta|\mathcal{D})$ by $\overline\Sigma_\mathcal{D}$, and $\overline\Sigma_S$ that obtained under $\widetilde\pi(\theta|S)$.
            We are not aware of any results on the asymptotic variability of discrepancy-based posteriors. For this reason no result like Lemma~\ref{lem:gamstar} can be given in this setting.
            However, estimated pooling weights can still be constructed in the same manner as the summary-based case. Namely, we can pool posteriors using the weights $\widehat\omega$ and $\widetilde\omega$ given earlier, which yields
			$$
			\widehat\omega_{}=1-\frac{\tr \overline\Sigma_\mathcal{D}}{\tr \{\overline\Sigma_\mathcal{D}+\overline\Sigma_S\}},\quad\text{ and }\quad \widetilde\omega_{}=1-\frac{\tr \overline\Sigma_\mathcal{D}}{(\bar\theta_S-\bar\theta_{\mathcal{D}})^\top(\bar\theta_S-\bar\theta_{\mathcal{D}})+\tr \{\overline\Sigma_{\mathcal{D}}+\overline\Sigma_S\}},
			$$where $\bar\theta_S$ denotes the posterior mean of $\widetilde\pi(\theta|S)$, and $\bar\theta_{\mathcal{D}}$ the posterior mean of $\widetilde\pi(\theta|\mathcal{D})$. 
			
			In the case of standard LFI posteriors, it is not at all clear how to combine inferences based on summaries and general discrepancies. If one were to attempt to construct a combined distance over the summaries and discrepancies, the resulting properties of such a combination are unknown, and presents issues also from a computational theoretical standpoint. In contrast, it is very simple to sample $\widetilde\pi(\theta|\mathcal{D})$, and $\widetilde\pi(\theta|S)$ separately, and fuse them together using $\widetilde\pi_\omega(\theta|\y)$. 
            The theoretical results derived in Section~\ref{sec:optpost} can, in principle, apply to combining posteriors built from summaries and discrepancies, so long as $\pi(\theta|\mathcal{D})$ satisfies a Bernstein-von Mises type results. However, the validity of such an assumption is not known for general choices of $\mathcal{D}$. We leave the extension of these results to the case of summaries and discrepancies for future research.
			
			\subsection{Examples: summaries and discrepancies}
			We now demonstrate the usefulness of this approach by combining posterior information built across combinations of summaries and discrepancies. In these experiments, we set $\mathcal{D}$ to be the Wasserstein metric, which yields the Wasserstein ABC (W-ABC) posterior studied in \cite{Bernton2017}; while other choices are entirely feasible, we maintain this choice as it is a popular metric. In addition, we conduct inference using BSL based on a generic auxiliary model; namely, we consider inference based on the summaries from a three component Gaussian mixture model. \cite{drovandi2021comparison} demonstrate that this choice performs well across several different experiments in terms of an accuracy comparison across many different likelihood-free approaches. 
			
			The choice of BSL for these experiments is deliberate and done to emphasize the practical usefulness of the pooling approach: in the case of BSL, it is not clear how to combine general discrepancies and summaries, since the form of the BSL posterior does not allow the incorporation of discrepancy distances.  
			
			\subsubsection{Example: g-and-k}
			For this experiment, we use precisely the same simulated data generated under the g-and-k model in Section \ref{sec:gandk}, and compare the results for BSL based on the auxiliary model summaries, against those obtained from the W-ABC approach, and the resulting pooled posteriors. We present all the same accuracy information as in Section \ref{sec:gandk} in Table \ref{tab:auxbasewass_gandk}. However, we note that in the experiments of \cite{drovandi2021comparison}, BSL coupled with the auxiliary model summaries performed very well, and so we would expect, \textit{a priori}, for the pooling weights to be close to unity across the experiments. 
			
			Analyzing Table \ref{tab:auxbasewass_gandk}, we see that the pooled posteriors have accuracy measures that are very similar to those obtained from the BSL posterior. The average weight on the BSL posterior under $\tilde{\omega}$ is $0.9817$. While not entirely surprising given the results of \cite{drovandi2021comparison}, the results demonstrate that posterior pooling is capable of providing large weight in cases where one set of information is clearly dominant. 
			\begin{table}[H]
				\centering
				{\footnotesize
					\begin{tabular}{lrrrrrrrrr}
						\hline\hline
						\textbf{(A)} & \multicolumn{3}{c}{$S$} & \multicolumn{3}{c}{$\mathcal{D}$} & \multicolumn{3}{c}{$(S,\mathcal{D})$} \\ 
						\textbf{} & Bias & Std    & MSE     & Bias      & Std     &MSE   & Bias & Std&MSE \\ \hline
						$a$ &0.0006   &   0.0373  &  0.0013 &  -0.0024 &   0.0444 &   0.0013 & --  & --    &  -- \\
						$b$ &0.0103  &  0.0767 &   0.0050  &  0.0161  &  0.0848  &  0.0050 &   --  &  --   &  --  \\
						$g$ &0.0200  &  0.1360 &   \textbf{0.0193}  &  0.0712 &   0.2410  &  0.0400 &   --  & --    & --  \\
						$k$ &-0.0071 &   0.0461 &   \textbf{0.0019} &  -0.0101&    0.0587 &   0.0023 &  --  & --    &  --   \\\hline\hline
						\textbf{(B)} & \multicolumn{3}{c}{$\omega=1/2$} & \multicolumn{3}{c}{$\widehat\omega$} & \multicolumn{3}{c}{$\widetilde\omega_{}$}   \\ 
						\textbf{} & Bias & Std & MSE& Bias & Std &MSE& Bias & Std&MSE \\ \hline
						$a$ &	  -0.0009 &   0.0413 &   \textbf{0.0012} &  -0.0001 &   0.0396&    \textbf{0.0012}&    0.0006&    0.0374 &   0.0013\\
						$b$ &	0.0132   & 0.0814  &  \textbf{0.0049}  &  0.0121  &  0.0795  &  \textbf{0.0049}   & 0.0103  &  0.0768  &  0.0050\\
						$g$&	0.0456  &  0.2024  &  0.0259  &  0.0311  &  0.1734  &  0.0196   & 0.0201  &  0.1379  &  \textbf{0.0193}\\
						$k$&	-0.0086&    0.0541 &   0.0020 &  -0.0083  &  0.0512 &   \textbf{0.0019}  & -0.0072 &   0.0464  &  \textbf{0.0019}\\\hline\hline
				\end{tabular}}
				\caption{Pooled posterior accuracy results in the g-and-k model under summaries ($S$) and discrepancies ($\mathcal{D}$). The remaining information is as in Table \ref{tab:robbase_oct}. The overall MSE over the replications is: $S$: 0.0275 ; $\mathcal{D}$: 0.0486; $\omega=1/2$: 0.0341; $\widehat\omega_{}$: {0.0295}; $\widetilde\omega_{}$: \textbf{ 0.0275}}.
				\label{tab:auxbasewass_gandk}
			\end{table}

			\subsubsection{Example: M/G/1}\label{sec:mg1}
			Next we consider an M/G/1 queueing model, which is a stochastic single-server queue model with Poisson arrivals and a general service time distribution. We follow existing constructions of this model in the LFI literature and maintain that the service times are $\mathcal{U}(\theta_1,\theta_2)$ (see e.g. \citealp{an2020robust}), while we consider that the inter-arrival times are distributed as $\text{Exp}(\theta_3)$.  We take the observed data $\y$ to be the inter-departure times of 51 customers, resulting in 50 observations. We generate 100 synthetic datasets from this model according to the true parameters  $(\theta_1,\theta_2,\theta_3)^\top = (1,5,0.2)^\top$. Since the service times are uniformly distributed we have the natural constraint that $\theta_1 < \min(y_1,y_2,\ldots,y_n)$ and so we incorporate that in the prior. Our prior beliefs on $(\theta_1,\theta_2,\theta_3)$ are thus given by $\mathcal{U}(0,\min(y_1,y_2,\ldots,y_n))\times \mathcal{U}(0,10+\min(y_1,y_2,\ldots,y_n)) \times \mathcal{U}(0,0.5)$.  
			
			The summaries used in this example (denoted by $S$) are again those based on an auxiliary Gaussian mixture (three components, so that $\mbox{dim}(S) = 8$), and the discrepancy used (denoted $\mathcal{D}$) is the 1-Wasserstein distance. In the experiments of \cite{Bernton2017}, the W-ABC posterior was shown to perform well against various summary-based counterparts, but in the experiments of \cite{drovandi2021comparison} the BSL posterior based on $S$ performed just as well as the W-ABC posterior. Thus, we expect the pooling weights between the two posteriors to be non-trivial.

			The results across the synthetic datasets are presented in Table \ref{tab:auxbasewass_mg1}, and demonstrate that there are (again) appreciable gains to be obtained by using pooled posteriors. In this experiment, using the pooled posterior based on $\widetilde\omega$ produces a 27\% reduction in the risk relative to using the BSL posterior alone, and a 25\% reduction in risk relative to using the W-ABC posterior. 
			
			\begin{table}[H]
				\centering
				{\footnotesize
					\begin{tabular}{lrrrrrrrrr}
						\hline\hline
						\textbf{(A)} & \multicolumn{3}{c}{$S$} & \multicolumn{3}{c}{$\mathcal{D}$} & \multicolumn{3}{c}{$(S,\mathcal{D})$} \\ 
						\textbf{} & Bias       & Std & MSE& Bias & Std &MSE& Bias & Std&MSE \\ \hline
						$\theta_1$ & -0.0627  &  0.1750 &   \textbf{0.0189}&   -0.1040 &   0.1845 &   0.0230& --  & --    &--   \\
						$\theta_2$ &0.0801    &  0.7260  &  0.5364   & 0.3051 &   0.8476   & 0.5212 & --    &  --   & --   \\
						$\theta_3$ &0.0598    &  0.0281 &   \textbf{0.0051}  &  0.0624 &   0.0332  &  0.0056 & --   & --    &--     \\\hline\hline
						\textbf{(B)} & \multicolumn{3}{c}{$\omega=1/2$} & \multicolumn{3}{c}{$\widehat\omega$} & \multicolumn{3}{c}{$\widetilde\omega_{}$}   \\ 
						$\theta_1$ &-0.0834  &  0.1839 &   0.0198 &  -0.0787 &   0.1815 &   0.0194 &  -0.0667 &   0.1766 &   0.0190\\
						$\theta_2$ &0.1926 &   0.8383  &  0.4451&    0.1684 &   0.7898 &   \textbf{0.4147} &   0.0924 &   0.7384  &  0.5198\\
						$\theta_3$ &0.0611 &   0.0324  &  0.0052 &   0.0611  &  0.0314&    0.0052  &  0.0605 &   0.0294  &  \textbf{0.0051}\\\hline\hline
				\end{tabular}}
				\caption{Pooled posterior accuracy results in the M/G/1 model under summaries ($S$) and discrepancies ($\mathcal{D}$). The remaining information is as in Table \ref{tab:robbase_oct}. The overall MSE over the replications is: $S$: 0.5603 ; $\mathcal{D}$: 0.5498; $\omega=1/2$: 0.4701; $\widehat\omega_{}$: 0.5439; $\widetilde\omega_{}$: \textbf{ 0.4393}.}
				
				\label{tab:auxbasewass_mg1}
			\end{table}
			
			\section{Discussion}\label{sec:discussion}
			
			In this work we propose to combine LFI posteriors based on different summary statistics, 
			or based on summary statistics and general discrepancy measures.
			A linear opinion pool of the component LFI posteriors is used 
			for the combination, and under appropriate assumptions 
			improved performance can be
			achieved for the pooled posterior mean in terms of asymptotic
			frequentist risk. {Additionally, if one of the summaries used is incompatible, we demonstrate that the corresponding component of the pool will receive zero weight asymptotically. Hence, not only can this pooled posterior improve point estimation compared to the individual LFI posteriors, but it can also guard against the impacts of model incompatibility in LFI, see, e.g., \cite{frazier2020model} and \cite{frazier2021synthetic} for details.}  
			
			While we consider linear pools to combine different LFI posterios, considering alternative strategies such as non-linear pools or ensemble methods such as stacking and bagging could be an interesting direction for future research. Looking to future work, it is of interest to apply similar methods
			in the context of modular posterior inferences for LFI 
			\citep{chakra2022modular}.  For discussions of modular Bayesian
			inference outside the LFI context see 
			\cite{liu+bb09}, \cite{lunn+bsgn09}, \cite{plummer15}, \cite{jacob+mhr17}
			and \cite{carmona+n20}.
			In \cite{chakra2022modular}, the authors consider a misspecified model
			and marginal inferences for a subset $\varphi$ of the parameters $\theta$.  
			They consider a linear opinion pool as a pooled posterior for $\varphi$, 
			with component LFI posteriors employing summary statistics $S_1$ and
			$S_2$, $S_1\subset S_2$.  The summaries $S_1$ are chosen
			to deliver reliable but possibly imprecise inferences about $\varphi$, 
			whereas $S_2$ can deliver more precise inferences, which we
			feel nevertheless should not be trusted if they are in conflict with
			the inferences derived from $S_1$.  
			
			The theory developed here must be modified in the case where
			$S_1\subset S_2$, or more generally where a joint core set of summaries appears in both $S_1$ and $S_2$.  In particular, Assumptions 1 and 2 in Section 3, 
			which assume that $S=(S_1^\top,S_2^\top)^\top$ has a strictly
			positive definite limiting covariance matrix under both the true
			data generating process and under the model, do not hold in
			this situation.  However, perhaps a more significant difficulty
			is that if the dimension of $S_2$ is much higher than $S_1$, it becomes
			more delicate to take the different levels of Monte Carlo error
			in the component LFI posteriors
			into account in the estimation of an appropriate mixing weight.
			
\begin{acks}[Acknowledgments]
	David Frazier was supported by the Australian Research Council's Discovery Early Career Researcher Award funding scheme (DE200101070). Christopher Drovandi was supported by the Australian Research Council Future Fellowships Scheme (FT210100260).  David Nott's research was supported by the Ministry of Education, Singapore, under the Academic Research Fund Tier 2 (MOE-T2EP20123-0009) and he is affiliated with the Institute of Operations Research and Analytics at the National University of Singapore.  
\end{acks}
			
			{ 
				{\footnotesize
					\bibliography{refs-2}

\begin{thebibliography}{}

\bibitem[An et~al., 2020]{an2020robust}
An, Z., Nott, D.~J., and Drovandi, C. (2020).
\newblock Robust {B}ayesian synthetic likelihood via a semi-parametric approach.
\newblock {\em Statistics and Computing}, 30(3):543--557.

\bibitem[An et~al., 2022]{JSSv101i11}
An, Z., South, L.~F., and Drovandi, C. (2022).
\newblock {BSL}: An {R} package for efficient parameter estimation for simulation-based models via {B}ayesian synthetic likelihood.
\newblock {\em Journal of Statistical Software}, 59(11):1–33.

\bibitem[Andersen and S{\o}rensen, 1996]{andersen1996gmm}
Andersen, T.~G. and S{\o}rensen, B.~E. (1996).
\newblock {GMM} estimation of a stochastic volatility model: A {M}onte {C}arlo study.
\newblock {\em Journal of Business \& Economic Statistics}, 14(3):328--352.

\bibitem[Andrews, 1991]{andrews1991heteroskedasticity}
Andrews, D.~W. (1991).
\newblock Heteroskedasticity and autocorrelation consistent covariance matrix estimation.
\newblock {\em Econometrica: Journal of the Econometric Society}, (3):817--858.

\bibitem[Ariely et~al., 2000]{AriTunBenBudDieHonWalZau2000}
Ariely, D., Tung~Au, W., Bender, R.~H., Budescu, D.~V., Dietz, C.~B., Gu, H., Wallsten, T.~S., and Zauberman, G. (2000).
\newblock The effects of averaging subjective probability estimates between and within judges.
\newblock {\em Journal of Experimental Psychology: Applied}, 6(2):130.

\bibitem[Bernton et~al., 2019]{Bernton2017}
Bernton, E., Jacob, P.~E., Gerber, M., and Robert, C.~P. (2019).
\newblock Approximate {B}ayesian computation with the {W}asserstein distance.
\newblock {\em Journal of the Royal Statistical Society Series B}, 81(2):235--269.

\bibitem[Blum, 2010]{blum2010non}
Blum, M.~G. (2010).
\newblock Approximate {B}ayesian computation: A nonparametric perspective.
\newblock {\em Journal of the American Statistical Association}, 105(491):1178--1187.

\bibitem[Blum et~al., 2013]{blum2013}
Blum, M.~G., Nunes, M.~A., Prangle, D., and Sisson, S.~A. (2013).
\newblock A comparative review of dimension reduction methods in approximate {B}ayesian computation.
\newblock {\em Statistical Science}, 28(2):189--208.

\bibitem[Browning et~al., 2018]{Browning2018}
Browning, A.~P., McCue, S.~W., Binny, R.~N., Plank, M.~J., Shah, E.~T., and Simpson, M.~J. (2018).
\newblock Inferring parameters for a lattice-free model of cell migration and proliferation using experimental data.
\newblock {\em Journal of Theoretical Biology}, 437:251--260.

\bibitem[Carmona and Nicholls, 2020]{carmona+n20}
Carmona, C. and Nicholls, G. (2020).
\newblock Semi-modular inference: enhanced learning in multi-modular models by tempering the influence of components.
\newblock In Chiappa, S. and Calandra, R., editors, {\em Proceedings of the Twenty Third International Conference on Artificial Intelligence and Statistics}, volume 108 of {\em Proceedings of Machine Learning Research}, pages 4226--4235. PMLR.

\bibitem[Chakraborty et~al., 2022]{chakra2022modular}
Chakraborty, A., Nott, D.~J., Drovandi, C., Frazier, D.~T., and Sisson, S.~A. (2022).
\newblock Modularized {B}ayesian analyses and cutting feedback in likelihood-free inference.
\newblock {\em arXiv preprint arXiv:2203.09782}.

\bibitem[Chen et~al., 2021]{Chen2021a}
Chen, Y., Zhang, D., Gutmann, M.~U., Courville, A., and Zhu, Z. (2021).
\newblock Neural approximate sufficient statistics for implicit models.
\newblock In {\em International Conference on Learning Representations (ICLR)}. arXiv:2010.10079.

\bibitem[Drovandi and Frazier, 2022]{drovandi2021comparison}
Drovandi, C. and Frazier, D.~T. (2022).
\newblock A comparison of likelihood-free methods with and without summary statistics.
\newblock {\em Statistics and Computing}, 32:42.

\bibitem[Drovandi and Pettitt, 2011a]{drovandi2011estimation}
Drovandi, C.~C. and Pettitt, A.~N. (2011a).
\newblock Estimation of parameters for macroparasite population evolution using approximate {B}ayesian computation.
\newblock {\em Biometrics}, 67(1):225--233.

\bibitem[Drovandi and Pettitt, 2011b]{drovandi2011likelihood}
Drovandi, C.~C. and Pettitt, A.~N. (2011b).
\newblock Likelihood-free {B}ayesian estimation of multivariate quantile distributions.
\newblock {\em Computational Statistics \& Data Analysis}, 55(9):2541--2556.

\bibitem[Evans and Guo, 2022]{evans2022combining}
Evans, M. and Guo, Y.~J. (2022).
\newblock Combining evidence.
\newblock {\em arXiv preprint arXiv:2202.02922}.

\bibitem[Fearnhead and Prangle, 2012]{FP2012}
Fearnhead, P. and Prangle, D. (2012).
\newblock Constructing summary statistics for approximate {B}ayesian computation: semi-automatic approximate {B}ayesian computation.
\newblock {\em Journal of the Royal Statistical Society: Series B (Statistical Methodology)}, 74(3):419--474.

\bibitem[Frazier and Drovandi, 2021]{frazier2021robust}
Frazier, D.~T. and Drovandi, C. (2021).
\newblock Robust approximate {B}ayesian inference with synthetic likelihood.
\newblock {\em Journal of Computational and Graphical Statistics}, 30(4):958--976.

\bibitem[Frazier et~al., 2018]{FMRR2016}
Frazier, D.~T., Martin, G.~M., Robert, C.~P., and Rousseau, J. (2018).
\newblock Asymptotic properties of approximate {B}ayesian computation.
\newblock {\em Biometrika}, 105(3):593--607.

\bibitem[Frazier et~al., 2024]{frazier2021synthetic}
Frazier, D.~T., Nott, D.~J., and Drovandi, C. (2024).
\newblock Synthetic likelihood in misspecified models.
\newblock {\em Journal of the American Statistical Association}, 0(0):1--12.

\bibitem[Frazier et~al., 2022]{frazier2019bayesian}
Frazier, D.~T., Nott, D.~J., Drovandi, C., and Kohn, R. (2022).
\newblock Bayesian inference using synthetic likelihood: asymptotics and adjustments.
\newblock {\em Journal of the {A}merical {S}tatistical {A}ssociation}, page in press.

\bibitem[Frazier et~al., 2019]{frazier2019indirect}
Frazier, D.~T., Oka, T., and Zhu, D. (2019).
\newblock Indirect inference with a non-smooth criterion function.
\newblock {\em Journal of Econometrics}, 212(2):623--645.

\bibitem[Frazier et~al., 2020]{frazier2020model}
Frazier, D.~T., Robert, C.~P., and Rousseau, J. (2020).
\newblock Model misspecification in approximate {B}ayesian computation: consequences and diagnostics.
\newblock {\em Journal of the Royal Statistical Society: Series B (Statistical Methodology)}.

\bibitem[Jacob et~al., 2017]{jacob+mhr17}
Jacob, P.~E., Murray, L.~M., Holmes, C.~C., and Robert, C.~P. (2017).
\newblock Better together? {S}tatistical learning in models made of modules.
\newblock arXiv:1708.08719.

\bibitem[Jiang, 2018]{Jia2018}
Jiang, B. (2018).
\newblock Approximate {B}ayesian computation with {Kullback-Leibler} divergence as data discrepancy.
\newblock In Storkey, A. and Perez-Cruz, F., editors, {\em Proceedings of the Twenty-First International Conference on Artificial Intelligence and Statistics}, volume~84 of {\em Proceedings of Machine Learning Research}, pages 1711--1721. PMLR.

\bibitem[Joyce and Marjoram, 2008]{joyce2008approximately}
Joyce, P. and Marjoram, P. (2008).
\newblock Approximately sufficient statistics and {B}ayesian computation.
\newblock {\em Statistical Applications in Genetics and Molecular Biology}, 7(1).

\bibitem[Kelly et~al., 2023]{kelly2023misspecification}
Kelly, R.~P., Nott, D.~J., Frazier, D.~T., Warne, D.~J., and Drovandi, C. (2023).
\newblock Misspecification-robust sequential neural likelihood.
\newblock {\em arXiv preprint arXiv:2301.13368}.

\bibitem[Lehmann and Casella, 2006]{lehmann2006theory}
Lehmann, E.~L. and Casella, G. (2006).
\newblock {\em Theory of point estimation}.
\newblock Springer Science \& Business Media.

\bibitem[Li and Fearnhead, 2018]{LF2016}
Li, W. and Fearnhead, P. (2018).
\newblock On the asymptotic efficiency of approximate {B}ayesian computation estimators.
\newblock {\em Biometrika}, 105(2):285--299.

\bibitem[Liu and Berger, 2009]{liu+bb09}
Liu, F.~Bayarri, M.~J. and Berger, J.~O. (2009).
\newblock Modularization in {B}ayesian analysis, with emphasis on analysis of computer models.
\newblock {\em Bayesian Analysis}, 4(1):119--150.

\bibitem[Lunn et~al., 2009]{lunn+bsgn09}
Lunn, D., Best, N., Spiegelhalter, D., Graham, G., and Neuenschwander, B. (2009).
\newblock Combining {MCMC} with `sequential' {PKPD} modelling.
\newblock {\em Journal of Pharmacokinetics and Pharmacodynamics}, 36:19--38.

\bibitem[Marchand et~al., 2017]{marchand2017stochastic}
Marchand, P., Boenke, M., and Green, D.~M. (2017).
\newblock A stochastic movement model reproduces patterns of site fidelity and long-distance dispersal in a population of {F}owler's toads ({A}naxyrus fowleri).
\newblock {\em Ecological Modelling}, 360:63--69.

\bibitem[Marin et~al., 2014]{Marinetal2014}
Marin, J.-M., Pillai, N.~S., Robert, C.~P., and Rousseau, J. (2014).
\newblock Relevant statistics for {B}ayesian model choice.
\newblock {\em Journal of the Royal Statistical Society: Series B (Statistical Methodology)}, 76(5):833--859.

\bibitem[Marin et~al., 2012]{marinea2012}
Marin, J.-M., Pudlo, P., Robert, C.~P., and Ryder, R.~J. (2012).
\newblock Approximate {B}ayesian computational methods.
\newblock {\em Statistics and Computing}, 22(6):1167--1180.

\bibitem[Martin et~al., 2019]{martin2019auxiliary}
Martin, G.~M., McCabe, B.~P., Frazier, D.~T., Maneesoonthorn, W., and Robert, C.~P. (2019).
\newblock Auxiliary likelihood-based approximate {B}ayesian computation in state space models.
\newblock {\em Journal of Computational and Graphical Statistics}, 28(3):508--522.

\bibitem[McAndrew and Reich, 2022]{McaRei2022}
McAndrew, T. and Reich, N.~G. (2022).
\newblock An expert judgment model to predict early stages of the {COVID-19} pandemic in the united states.
\newblock {\em PLoS Computational Biology}, 18(9):e1010485.

\bibitem[Nguyen et~al., 2020]{NguArbJulLueFor2020}
Nguyen, H.~D., Arbel, J., L{\"u}, H., and Forbes, F. (2020).
\newblock Approximate {B}ayesian computation via the energy statistic.
\newblock {\em IEEE Access}, 8:131683--131698.

\bibitem[Nott et~al., 2020]{nott2020checking}
Nott, D.~J., Wang, X., Evans, M., and Englert, B.-G. (2020).
\newblock Checking for prior-data conflict using prior-to-posterior divergences.
\newblock {\em Statistical Science}, 35(2):234--253.

\bibitem[Nunes and Balding, 2010]{nunes2010optimal}
Nunes, M.~A. and Balding, D.~J. (2010).
\newblock On optimal selection of summary statistics for approximate {B}ayesian computation.
\newblock {\em Statistical Applications in Genetics and Molecular Biology}, 9(1).

\bibitem[Plummer, 2015]{plummer15}
Plummer, M. (2015).
\newblock Cuts in {B}ayesian graphical models.
\newblock {\em Statistics and Computing}, 25:37--43.

\bibitem[Politis and Romano, 1994]{politis1994stationary}
Politis, D.~N. and Romano, J.~P. (1994).
\newblock The stationary bootstrap.
\newblock {\em Journal of the American Statistical association}, 89(428):1303--1313.

\bibitem[Prangle, 2018]{prangle2015summary}
Prangle, D. (2018).
\newblock Summary statistics.
\newblock In {\em Handbook of approximate Bayesian computation}, pages 125--152. Chapman and Hall/CRC.

\bibitem[Price et~al., 2018]{price2018bayesian}
Price, L.~F., Drovandi, C.~C., Lee, A., and Nott, D.~J. (2018).
\newblock Bayesian synthetic likelihood.
\newblock {\em Journal of Computational and Graphical Statistics}, 27(1):1--11.

\bibitem[Priddle et~al., 2022]{priddle2019efficient}
Priddle, J.~W., Sisson, S.~A., Frazier, D.~T., and Drovandi, C. (2022).
\newblock Efficient {B}ayesian synthetic likelihood with whitening transformations.
\newblock {\em Journal of Computational and Graphical Statistics}, 31(1):50--63.

\bibitem[Rayner and MacGillivray, 2002]{rayner2002numerical}
Rayner, G.~D. and MacGillivray, H.~L. (2002).
\newblock Numerical maximum likelihood estimation for the g-and-k and generalized g-and-h distributions.
\newblock {\em Statistics and Computing}, 12(1):57--75.

\bibitem[Sisson and Fan, 2011]{sisson2011likelihood}
Sisson, S.~A. and Fan, Y. (2011).
\newblock Likelihood-free {MCMC}.
\newblock {\em Handbook of Markov Chain Monte Carlo}, pages 313--335.

\bibitem[Sisson et~al., 2018]{sisson2018handbook}
Sisson, S.~A., Fan, Y., and Beaumont, M. (2018).
\newblock {\em Handbook of Approximate {B}ayesian Computation}.
\newblock Chapman and Hall/CRC, New York.

\bibitem[Stone, 1961]{stone1961opinion}
Stone, M. (1961).
\newblock The opinion pool.
\newblock {\em The Annals of Mathematical Statistics}, pages 1339--1342.

\bibitem[Tejero-Cantero et~al., 2020]{Tejero-Cantero2020}
Tejero-Cantero, A., Boelts, J., Deistler, M., Lueckmann, J.-M., Durkan, C., Gon{\c c}alves, P.~J., Greenberg, D.~S., and Macke, J.~H. (2020).
\newblock sbi: A toolkit for simulation-based inference.
\newblock {\em Journal of Open Source Software}, 5(52):2505.

\bibitem[Wang et~al., 2022]{wang2022forecast}
Wang, X., Hyndman, R.~J., Li, F., and Kang, Y. (2022).
\newblock Forecast combinations: an over 50-year review.
\newblock {\em arXiv preprint arXiv:2205.04216}.

\bibitem[Wood, 2010]{wood2010statistical}
Wood, S.~N. (2010).
\newblock Statistical inference for noisy nonlinear ecological dynamic systems.
\newblock {\em Nature}, 466(7310):1102--1104.

\bibitem[Yao et~al., 2023]{yao2023simulation}
Yao, Y., Blancard, B. R.-S., and Domke, J. (2023).
\newblock Simulation based stacking.
\newblock {\em arXiv preprint arXiv:2310.17009}.

\bibitem[Yao et~al., 2018]{yao+vsg18}
Yao, Y., Vehtari, A., Simpson, D., and Gelman, A. (2018).
\newblock {Using Stacking to Average Bayesian Predictive Distributions (with Discussion)}.
\newblock {\em Bayesian Analysis}, 13(3):917 -- 1007.

\end{thebibliography}
					\bibliographystyle{apalike}}
			}

\appendix

\section{Additional Discussion and Examples}\label{sec:adddiss}
\subsection{Summary Selection}\label{sec:formalapp}

\cite{FP2012} consider the problem of choosing summaries by attempting to give a decision rule $\delta\in\Theta$ that minimises the posterior expected loss
\begin{flalign}
	R_S(\delta)=\int (\theta-\delta)^\top  (\theta-\delta)
	\widetilde\pi(\theta|S)\dt\theta.
\end{flalign}Viewing the above loss as a function of arbitrary summaries $S$, \cite{FP2012} argue that taking $S=\E[\theta|S(\by)]$ results in minimizing $R_S(\delta)$. To estimate this summary statistic \cite{FP2012} propose the use of (non)linear regression methods on (powers of) the summary statistics. That is, given training data $\{\theta,S(\bz)\}$, \cite{FP2012} use as a summary statistic the fitted regression function evaluated at $S(\by)$. 

The goal of \cite{FP2012}, is not to choose between summaries, but to approximate the most informative collection given a fixed set of summaries $S$. In this way, there is no sense in which the use of posterior expected loss should deliver a helpful criterion for deciding amongst competing collections of summaries $S_1$ and $S_2$. Indeed, the minimum of $R_{S_j}(\delta)$ is obtained at $\delta=\bar\theta_j=\int \theta \widetilde\pi(\theta|S_j)\dt\theta$,  for each  $S_j$, and, under quadratic loss, $R_{S_j}(\bar\theta_j)\approx\tr\left[B_j^\top V_j^{-1} B_j\right]^{-1}/n$ (i.e., the asymptotic variance of the posterior). Consequently, choosing summaries by comparing $R_{S_j}(\bar\theta_j)$ would lead us to choose whichever collection of summaries delivered the smallest posterior variance. This is not a helpful selection criterion since, under weak regularity conditions, the posterior variance of $\pi(\theta|S)$ is (asymptotically) a decreasing function of the number of summaries in $S$; that is, asymptotically, adding more summaries can never increase the LFI posterior variance (see, e.g., \citealp{FMRR2016} for theoretical justification of this claim).  

Consequently, according to posterior expected loss, $S$ would asymptotically produce the smallest loss. The latter decision rule is unhelpful in practice as it completely disregards the fact that the computational resources required to approximate the posterior can increase drastically as the dimension of the summaries increase. While those resources are somewhat mitigated if one considers the \cite{FP2012} approach, it remains that the use of $R_{S_j}(\bar\theta_j)$ completely ignores the difference between posterior locations that arises when using different summaries, i.e., in general $\E_{\pi(\theta|S_1)}(\theta)\ne\E_{\pi(\theta|S_2)}(\theta)\ne \E_{\pi(\theta|S)}(\theta).$  When $S_1(y_{obs})\ne\E(\theta\mid y_{obs})$ and $S_2(y_{obs})\ne\E(\theta\mid y_{obs})$, the criterion $R_S(\delta)$ does not deliver a meaningful way to choose between $S_1$ or $S_2$: the loss $R(\delta)$ would simply choose whichever grouping of summaries was closest to $\E[\theta\mid y_{obs}]$, and would not deliver a way of trading off between $S_1$ and $S_2$.

The asymptotic viewpoint also masks the critically important issue of the (Monte Carlo) accuracy of the resulting posterior approximation for a fixed computational budget. That is, while we can never decrease the asymptotic variance of the posterior by adding summaries, given a finite-time computational budget,  the variability of the posterior approximation is an increasing function of the dimension of the summaries (see, e.g., \citealp{blum2010non}). Hence, given a finite computational budget, a posterior approximation based on $S$, denoted by $\widetilde\pi(\theta|S)$, can easily have larger amounts of variability than a posterior approximation that targets a lower-dimensional set of summaries, e.g., $\widetilde\pi(\theta|S_1)$. Hence, in practice large collections of summary statistics are not generally helpful for LFI without the application of adjustment procedures. 

\subsection{Estimating $\Omega_\Sigma$}\label{sec:estCov}
To estimate the pooling weight $\omega^\star_+$, we must construct an estimator of the covariance matrix $\Omega_\Sigma = Q_1\Omega_{1,2}Q_2^\top$, where $Q_j=\Sigma_jB_j^\top V_j^{-1}$ for $j=1,2$; please see Section~3.1 of the main paper for specific definitions. Therefore, to estimate $\Omega_\Sigma$ we must estimate both the covariance matrix of the summaries, defined as $V_\Sigma$, and the gradient terms $B_j=\nabla_\theta b_j(\theta_0)$. 

The covariance matrix of the summaries can be estimated using bootstrapping methods. In particular, if the data is iid, we generate $b=1,\dots,B$, bootstrap replicates of the summary statistics $S^{(b)}=(S_1(\y^{(b)})^\top,S_2(\y^{(b)})^\top)^\top$, and then form the sample covariance matrix of the replicates $\overline{V}_\Sigma$, which itself is composed of the variance estimates $\overline{V}_j$ and covariance estimate $\overline{\Omega}_{1,2}$. In the numerical experiments with iid data - the g-and-k example - we used $B=1000$ replications. 

When the data is weakly dependent, a block-bootstrap can be employed to generate the corresponding bootstrapped samples, and a heteroskedastic and auto-correlation (HAC) consistent covariance matrix estimator used, e.g., \cite{andrews1991heteroskedasticity}, in place of the usual sample variance. For the stochastic volatility example, we used the block bootstrap of \cite{politis1994stationary} with block length of ten observations, and the standard sample covariance was applied to the bootstrapped summaries. In the stochastic volatility example, we found that the use of HAC variance matrix made no discernible difference to the results, and so used the simpler version. 

The last component needed to estimate $\Omega_\Sigma$ are the gradients $B_j=\nabla_\theta b_j(\theta_0)$. These components can be obtained through automatic or numerical differentiation using a large number of replications from the DGP generated under the posterior mean $\bar\theta_j:=\E_{\widetilde\pi(\theta|S_j)}[\theta]$. In particular, we can generate $N$ sample paths under $\bar\theta_j$ to obtain simulated data $\{\bar{z}_j\sim P_{\bar\theta_j}:j=1,\dots,N\}$, and estimate the gradient $B_j$ by differentiating the sample average 
$$
\overline{S}_j(\bar\theta_j):=N^{-1}\sum_{j=1}^{N}S_j(\bar{z}_j),
$$ with respect to each component of $\theta$. This can be done using automatic differentiation methods, or numerical differentiation. For instance, if $\theta\in\mathbb{R}$, then a  central finite difference  numerical derivative could be used to estimate $B_j$, whereby, for some small $h>0$, we simulate $\{\bar{z}_j^+\sim P_{\bar\theta_j+h}:j=1,\dots,N\}$ and $\{\bar{z}_j^{-}\sim P_{\bar\theta_j-h}:j=1,\dots,N\}$, and then estimate $B_j$ using 
$$
\overline{B}_j:=\frac{\overline{S}_j(\bar\theta_j+h)-\overline{S}_j(\bar\theta_j-h)}{2h}=\frac{1}{2hN}\sum_{j=1}^{N}\{S_j(\bar{z}_j^+)-S_j(\bar{z}_j^-)\}.
$$In the examples in the paper, the above central finite difference estimator was used with $N=1000$ sample paths. We also note that if the summaries are not particularly smooth in the unknown parameters, such as the case of quantiles of the data, the methods developed in \cite{frazier2019indirect} can be used to consistently estimate these components.

The posterior variance $\Sigma_j$ can be directly estimated using the sample variance of posterior draws from $\widetilde\pi(\theta\mid S_j)$, and we denote this estimate as $\overline\Sigma_j$. Given the estimators $\overline{B}_j$, $\overline{V}_\Sigma$, and $\overline\Sigma_j$, $Q_j$ can be estimated using $\overline{Q}_j=\overline\Sigma_j\overline{B}_j^\top\overline{V}_j^{-1}$, and $\Omega_\Sigma$ is then estimated as 
$
\overline\Omega_\Sigma=\overline{Q}_1\overline\Omega_{1,2}\overline{Q}_2.
$

\subsection{Cell Biology Example} \label{sec:cell}

Here we consider the lattice-free collective cell spreading model of \citet{Browning2018}.  The model permits cells to move freely in continuous space. There are three parameters in the model. There are two parameters that impact the spatial distribution of the cells, $m$ and $\gamma_b$.  The parameter $p$ affects the number of cells.  For specific details on the stochastic model, see \citet{Browning2018}.  

In the experiments of \citet{Browning2018}, images of the cell population are taken every 12 hours with the final image taken at 36 hours.  \citet{Browning2018} use the number of cells and the pair correlation computed from each of the three images as the summary statistics, resulting in a six dimensional summary statistic, $S$.  The pair correlation is the ratio of the number of pairs of agents separated by some pre-specified distance to an expected number of cells separated by the same distance if the cells were uniformly distributed in space.  In an attempt to learn more about $m$ and $\gamma_b$, \citet{priddle2019efficient} consider a higher dimensional set of statistics summarising the spatial information.  They consider Ripley's $K$ and $J$ functions evaluated at various diameters for each time point. Combined with the same total number of cells at each time point, there are 21 summary statistics in total;  see \citet{priddle2019efficient} for more details.  We refer to the two sets of summary statistics as ``pair" \citep{Browning2018} and ``spatial" \citet{priddle2019efficient}, respectively.  

We consider two likelihood-free algorithms.  One of them uses the SMC ABC replenishment algorithm of \citet{drovandi2011estimation} where the algorithm is stopped when the acceptance rate of the MCMC step falls below 1\%.  We also consider the MCMC BSL algorithm of \citet{price2018bayesian}.  For BSL, we use 10000 MCMC iterations with a random walk covariance matrix tuned using some pilot runs based on a simulated dataset.   Our results below are based on 50 independent datasets simulated using $m = 1$, $p = 0.04$ and $\gamma_b = 0.5$.   The prior distribution is set as $p \sim \mathcal{U}(0,10)$, $m \sim \mathcal{U}(0,0.2)$ and $\gamma_b \sim \mathcal{U}(0,20)$ with no dependence amongst parameters.

Firstly we consider pooling the results from ABC with the pair correlation statistics (ABC pair) and ABC with the spatial statistics (ABC spatial).  We might suspect that the spatial statistics will carry more information about $m$ and $\gamma_b$ than the pair statistics, but they have a higher dimension and we may be concerned that ABC spatial may produce inferences that are detrimental to $p$.  The results are shown in Table \ref{tab:abcpair_vs_abcspatial}.  It is evident that the pooled results improve on the inferences for $m$ and $\gamma_b$ compared to ABC pair (due to the good performance of ABC spatial for these two parameters) and improve on the inferences for $p$ compared to ABC spatial (due to the good performance of ABC pair for this parameter).  Since the parameter estimates are on different scales, we also consider pooling separately for each individual parameter.  We can see that pooling with the first set of weights produces low relative MSEs for all three parameters compared to the other approaches.

Secondly we consider pooling the results from ABC spatial and BSL spatial.  BSL avoids the tolerance error associated with ABC, but we might be concerned about its Gaussian likelihood assumption.  It turns out that BSL is very effective in this particular problem, since Table \ref{tab:abcspatial_vs_bslspatial} shows that it produces the smallest MSE for all parameters.  However, it can be seen that the pooled results produces small relative MSEs compared to ABC spatial.  Thus, there is only a small loss of efficiency compared to BSL spatial, whilst providing some robustness to the Gaussian assumption by pooling with the ABC results.

\begin{table}[H]
	\centering
	{\footnotesize
    		\begin{tabular}{lrrrrrr}
			\hline\hline
			\textbf{(A)} & \multicolumn{3}{c}{$S_1$} & \multicolumn{3}{c}{$S_2$}  \\ 
			\textbf{} & Bias & Std & MSE& Bias & Std &MSE \\ \hline
			$m$ & 0.0018  & 0.58   & 0.15   & 0.16   & 0.41    & 0.12    \\
			$p$ & -1.1e-4 &  0.0017  &  \textbf{2.9e-6}   & 3.2e-4  & 0.0030    & 4.1e-6 \\
			$\gamma_b$ & 2.8  & 4.09   & 10  & 0.35   & 1.45   & 1.1     \\\hline\hline
			\textbf{(B)} & \multicolumn{3}{c}{$\widehat\omega$} & \multicolumn{3}{c}{$\widetilde\omega$} \\ 
			\textbf{} & Bias & Std & MSE& Bias & Std &MSE \\ \hline
			$m$ & 0.17	   & 0.47    &  0.12   & 0.09  &  0.51    & 0.10  \\
			$p$ & 2.5e-4	   &  0.0028  & 3.8e-6 & 1.2e-4    & 0.0026   & 3.4e-6   \\
			$\gamma_b$&	 0.56  & 2.1  & \textbf{1.3}    & 1.4    &  3.1  & 3.1  \\\hline\hline
			\textbf{(C)} & \multicolumn{3}{c}{$\widehat\omega$} & \multicolumn{3}{c}{$\widetilde\omega$}   \\ 
			\textbf{} & Bias & Std & MSE& Bias & Std &MSE \\ \hline
			$m$ & 0.056	   & 0.49    &  \textbf{0.083}   & 0.039  &  0.51    & 0.10   \\
			$p$ & -1.6e-5	   &  0.0020  & 3.0e-6 & -2.2e-5    & 0.0020   & 3.0e-6 \\
			$\gamma_b$&	 0.54  & 2.1  & \textbf{1.3}    & 1.4    &  3.1  & 3.0  \\\hline\hline
	\end{tabular}}
	\caption{Pooled posterior accuracy results in the cell biology model under the base set of summaries $S_1$ (pair) and the alternative set $S_2$ (spatial). The remaining information is as in Table~1.  The average value of $\widehat\omega$ and $\widetilde\omega$ over the 50 datasets is 0.13 and 0.37, respectively, indicating preference for the inference based on the spatial summaries. (C) shows the same results as (B) except that pooling is done for each individual parameter to help remove the effect of scaling between different parameters. }  
\label{tab:abcpair_vs_abcspatial}
\end{table}

\begin{table}[H]
\centering
{\footnotesize
	\begin{tabular}{lrrrrrr}
		\hline\hline
		\textbf{(A)} & \multicolumn{3}{c}{ABC $S_2$} & \multicolumn{3}{c}{BSL $S_2$} \\ 
		\textbf{} & Bias & Std & MSE& Bias & Std &MSE \\ \hline
		$m$ & 0.16   &  0.41   &  0.12   & 0.10    & 0.21    &  \textbf{0.07}  \\
		$p$ & 3.2e-4  &  0.0030   &  4.1e-6   & 5.2e-6   &  0.0014    & \textbf{2.8e-6}    \\
		$\gamma_b$ &  0.35  &  1.4   & 1.1  &   -1.2e-4   &  0.59  & \textbf{0.23}    \\\hline\hline
		\textbf{(B)} & \multicolumn{3}{c}{$\widehat\omega$} & \multicolumn{3}{c}{$\widetilde\omega$}    \\ 
		\textbf{} & Bias & Std & MSE& Bias & Std &MSE  \\ \hline
		$m$ &  0.12	   &  0.27    &  \textbf{0.07}    & 0.12   &  0.29   &  0.08   \\
		$p$ &  4.6e-5  &  0.0017   & \textbf{2.8e-6}  &  1.0e-4   & 0.0020    & 3.1e-6  \\
		$\gamma_b$ & 0.03	   & 0.80   &   0.25   & 0.11     &  0.95   & 0.42   \\\hline\hline
\end{tabular}}
\caption{Pooled posterior accuracy results in the cell biology model under the base inference method (ABC spatial)  and the alternative inference method (BSL spatial). The remaining information is as in Table~1.  The average value of $\widehat\omega$ and $\widetilde\omega$ over the 50 datasets is 0.16 and 0.26, respectively, indicating preference for the BSL spatial results.}  
\label{tab:abcspatial_vs_bslspatial}
\end{table}

\subsection{Additional Results for the Examples}

Table~\ref{tab:average_weights} shows the average and standard deviation for the estimated weights across the repetitions for all experiments. 

\begin{table}[H]
\centering
{\footnotesize
	\begin{tabular}{lcccc}
\hline\hline
 & \multicolumn{2}{c}{$1-\widehat{\omega}$} & \multicolumn{2}{c}{$1-\widetilde{\omega}$} \\
 & mean & std & mean & std \\
 \hline
\textbf{(A)} Different types of summaries & \multicolumn{1}{l}{} & \multicolumn{1}{l}{} & \multicolumn{1}{l}{} & \multicolumn{1}{l}{} \\
g-and-k & 0.1289 & 0.0544 & 0.4925 & 0.1921 \\
Stochastic Volatility Model & 0.5902 & 0.1333 & 0.4925 & 0.1507 \\
\textbf{(B)} Different types of posteriors &  &  &  &  \\
g-and-k & 0.6975 & 0.0618 & 0.9817 & 0.0329 \\
M/G/1 & 0.5844 & 0.1353 & 0.8663 & 0.1328 \\
\hline\hline
\end{tabular}}
\caption{Average and standard deviation for the estimated weights $\widehat{\omega}$ and $\widetilde{\omega}$ across the repetitions for the experiments described in the main text. The table is orientated to show the weight put on the first individual posterior.} 
\label{tab:average_weights}
\end{table}

\section{Proofs of Main Results}\label{sec:proofs}

In this section, we prove the main results stated in the paper. However, before doing so, we state a few useful lemmas that allow us to simplify the proofs of certain results. 

\begin{proof}[Proof of Lemma~3]
The result follows from Assumption 1-3 and similar arguments to Corollary 1 in \cite{frazier2019bayesian}. In particular, following the arguments in Corollary 1 of \cite{frazier2019bayesian}, for $Z_{n,j}=Q_j\sqrt{n}\{S_j(\y)-b_j(\theta_0)\}$, 
\begin{flalign*}
\bar\theta_j &= \int \theta\widetilde\pi(\theta|S_j)\dt\theta= \int (\theta_0+t_j/\sqrt{n}+Z_{n,j}/\sqrt{n})\widetilde\pi(t|S_j)\dt t
\end{flalign*}
so that 
\begin{flalign*}
\sqrt{n}(\bar\theta_j-\theta_0)-Z_{n,j}&=\int t_j \widetilde\pi(t|S_j)\dt t\\&=\int t_j\{\widetilde\pi(t|S_j)-N(t_j;0,\Sigma_j)\}\dt t_j+\int t_jN(t_j;0,\Sigma_j)\dt t_j.
\end{flalign*}The second term is zero by definition, while the first term can be bounded as 
\begin{flalign*}
\int t_j\{\widetilde\pi(t|S_j)-N(t_j;0,\Sigma_j)\}\dt t_j\le \int \|t_j\||\{\widetilde\pi(t|S_j)-N(t_j;0,\Sigma_j)\}|\dt t_j=o_p(1)
\end{flalign*}where the $o_p(1)$ term follows by Assumption 3. 

Thus, it follows that  
\begin{flalign*}
&\sqrt{n}(\bar\theta_1-\theta_0)-Q_1\sqrt{n}\{S_1(\by)-b_1(\theta_0)\}=o_p(1)\\
&\sqrt{n}(\bar\theta_2-\theta_0)-Q_2^{}\sqrt{n}\{S_2(\by)-b_2(\theta_0)\}=o_p(1)
\end{flalign*}	
However, under Assumptions 1 and 2,
\begin{flalign*}
\sqrt{n}\{S_{2}(\by)-b_2(\theta_0)\}&=\sqrt{n}\{S_{2}(\by)-b_{2,0}\}+\sqrt{n}\{b_2(\theta_0)-b_{2,0}\}\\&=\sqrt{n}\{S_{2}(\by)-b_{2,0}\}+\sqrt{n}\delta_{2,n}\\&\Rightarrow \mathcal{N}\{\tau_2,V_2\},
\end{flalign*}where the second line follows from the convergence in Assumption 1. From the joint convergence of $S=(S_1^\top,S_2^\top)^\top$ in Assumption 1, the stated joint convergence then follows. 
\end{proof}

\begin{proof}[Proof of Lemma 4]
For $\bar\theta$ denoting $\bar\theta_1$ or $\bar\theta_2$, a second-order Taylor expansion of $L(\theta_0,\bar\theta)$ around $\theta_0$, with Lagrange remainder term $\vartheta$ satisfying $\|\vartheta-\theta_0\|\le C\|\theta-\theta_0\|$ for some $C>0$, yields
\begin{flalign*}
L(\theta_0,\bar\theta)&=L(\theta_0,\theta_0)+\partial L(\theta_0,\theta_0)/\partial\theta^\top (\theta-\theta_0)+\frac{1}{2}(\theta-\theta_0)^\top H(\theta_0)(\theta-\theta_0)\\&+\frac{1}{2}(\theta-\theta_0)^\top \left[H(\vartheta)-H(\theta_0) \right](\theta-\theta_0)\\&\le \frac{1}{2}\|(\bar\theta-\theta_0)\|_{H(\theta_0)}^2+M\|(\bar\theta-\theta_0)\|^3,
\end{flalign*}where the second line follows from Assumption 4 and the definition of the intermediate value. Hence, $$
nL(\theta_0,\bar\theta)=\frac{1}{2}\{\sqrt{n}(\bar\theta-\theta_0)\}^\top H(\theta_0)\{\sqrt{n}(\bar\theta-\theta_0)\}+o(\|\{\sqrt{n}(\bar\theta-\theta_0)\}\|^2)
$$

Define $Y_{j,n}:=\sqrt{n}(\bar\theta_j-\theta_0)$, and note that, by Lemma 3, $$Y_{j,n}\Rightarrow Y:=\begin{cases}N(0,\Sigma_1)&\text{ if }j=1\\
N(Q_2\tau_2,\Sigma_2)	&\text{ if }j=2
\end{cases}.$$ For $Q_{j,n}:=\|Y_{j,n}\|^2_H$, let $Y_{j,n,\zeta}=Y_{j,n}\I[Q_{j,n}\le\zeta]+\zeta \I[Q_{j,n}>\zeta]$. By Theorem 1.8.8 of \cite{lehmann2006theory}, 
$$
\lim_{n\rightarrow\infty}\E\left[\|Y_{j,n,\zeta}\|_{H(\theta_0)}^2\right]=\E\left[\|Y_j\|_{H_0}^2\I(\|Y_j\|_{H_0}^2\le \zeta)\right]+\zeta^2\text{Pr}(\|Y_j\|^2_{H_0}>\zeta).
$$For $\zeta\rightarrow\infty$, the RHS of the above converges to 
$$\E[\|Y_j\|^2_{H_0}]=\begin{cases}\tr H_0\Sigma_1&\text{ if }j=1\\\tau_2^\top{Q}_2^\top H_0Q_2\tau_2+\tr H_0\Sigma_2&\text{ if }j=2
\end{cases}.$$
\end{proof}

\begin{proof}[Proof of Lemma 1]
Write $$\sqrt{n}\{\bar\theta(\omega)-\theta_0\}=\sqrt{n}\{(1-\omega)\bar\theta_1+\omega\bar\theta_{2}-\theta_0\}=(1-\omega)\sqrt{n}\{\bar\theta_1-\theta_0\}+\omega\sqrt{n}\{\bar\theta_{2}-\theta_0\}.$$
Recall the definitions of $Q_j=\Sigma_j B_j^\top V_j^{-1}$, and let $Q^\star_1=[Q_1 \;\mathbf{0}_{d_\theta\times d_{s_2}}]$ and $Q^\star_2=[\mathbf{0}_{d_\theta\times d_{s_1}} \;Q_2]$. For $Z_{n}=\sqrt{n}\{S(\by)-b(\theta_0)\}$, under Assumptions 1, by Lemma 3, 
\begin{flalign*}
\sqrt{n}\{\bar\theta(\omega)-\theta_0\}&=(1-\omega)Q_1^\ast Z_{n}+\omega Q_2^\star Z_n \Rightarrow Y(\omega):=(1-\omega)Q_1^\star M+\omega Q_2^\star M
\end{flalign*}where $M\sim N(\xi,V_{1,2})$ with $\xi=(0^\top,\tau_2^\top)^\top$, and $V_{1,2}=\text{Var}[\sqrt{n}\{S(\by)-b(\theta_0)\}]$. 

Following similar arguments to the proof of Lemma 4 yields
$\mathcal{R}_0(\omega)=\E\left[\|Y(\omega)\|^2_{{H_0}}\|\right]$, and writing out $\|Y(\omega)\|^2_{H_0}$, we have 
\begin{flalign*}
\|Y(\omega)\|^2_{H_0}=(1-\omega)^2\|Q_1^\ast M\|^2_{H_0}+\omega^2\|Q_2^\ast M\|^2_{H_0}+2\omega(1-\omega) (Q_1^\ast M)^\top H_0  Q_2^\ast M .
\end{flalign*}The result follows by taking the expectations of each term, and solving for the optimal $\omega$.

For the first term, write $\|Q_1^\ast M\|^2_{H_0}=\|Q_1^\ast (M-\xi)+Q_1^\ast\xi\|^2_{H_0}$, and note that $Q_1^\ast\xi=0$. Hence, $$\E[\|Q_1^\ast (M-\xi)\|^2_{{H_0}}]=\tr {H_0}Q_1^\ast V_{1,2}(Q_1^\ast) ^\top=\tr{H_0}Q_1V_1Q_1^\top=\tr {H_0}[B_1^\top V_1^{-1}B_1]^{-1},$$where the last equality follows from the definition of $Q_1$. Applying a similar approach to the second term yields
\begin{flalign*}
\E\|Q_2^\ast M\|^2_{H_0}&=\tr {H_0} Q_2V_{2}Q_2^\top +\tau_2^\top Q_2^\top {H_0} Q_2\tau=\tr {H_0} [B_2^\top V_2^{-1} B_2]^{-1}+\tau_2^\top Q_2^\top {H_0} Q_2\tau
\end{flalign*}
For the last term, write 
\begin{flalign*}
\E\left[(Q_1^\ast M)^\top {H_0} \left\{Q_2^\ast M\right\}\right]&=\E\tr\left[ {H_0} \left\{Q_2^\ast M\right\}M^\top Q_1^{\ast\intercal}\right]\\&=\tr {H_0} Q_{2}^\ast\E[MM^\top] Q_1^{\ast\intercal}\\&=\tr {H_0} Q_2^\ast\{\xi\xi^\top +V_{1,2}\} Q_1^{\ast\intercal}\\&=\tr {H_0}Q_2\Omega_{2,1}Q_1^\top
\end{flalign*}where we have used $Q_1^\ast\xi=0$.

We recall the following notations: $\Sigma_1:=[B_1^\top V_{1}^{-1}B_1]^{-1}$ and $\Sigma_2=[B_{2}^\top V_{2}^{-1}B_{2}]^{-1}$. Using this, and the above expectations, $\mathcal{R}_0(\omega)$ becomes
\begin{flalign*}
\mathcal{R}_0(\omega)&=(1-\omega)^2\tr {H_0}\Sigma_1+\omega^2[\tr {H_0}\Sigma_2+\tau_2^\top Q_{2}^\top {H_0}Q_{2}\tau_2]+2\omega(1-\omega) \tr {H_0}\Omega_\Sigma\\&\equiv(1-\omega)^2\mathcal{R}_0(1)+\omega^2\mathcal{R}_0(2)+2\omega(1-\omega) \tr {H_0}\Omega_\Sigma
\end{flalign*}
To maximize $\mathcal{R}_0(\omega)$ over $\omega\in[0,1]$ we consider the Lagrangian 
$$
\mathcal{L}(\omega,\lambda)=\mathcal{R}_0(\omega)+\lambda(1-\omega),
$$where $\lambda$ is the multiplier associated to the constraint $(1-\omega)\ge0$.

First, consider that $\omega^\star$ is in the interior of the space, i.e., $0<\omega^\star<1$. Differentiating the above wrt $\omega$ and solving for $\omega$ as a function of $\lambda$ yields the solution:
\begin{equation}\label{eq:cond}
\omega^\star(\lambda)=\frac{\lambda}{2\mathcal{J}}+\frac{\mathcal{R}_0(1)-\tr {H_0}\Omega_\Sigma}{\mathcal{J}}=\frac{\lambda+2(\mathcal{R}_0(1)-\tr {H_0}\Omega_\Sigma)}{2\mathcal{J}},
\end{equation}
where $$\mathcal{J}:=\mathcal{R}_0(1)+\mathcal{R}_0(2)-\tr {H_0}\Omega_\Sigma=\tr {H_0}\left\{\Sigma_1+\Sigma_2-\Omega_\Sigma\right\}+\tau_2^\top Q_{2}^\top{H_0}Q_{2}\tau_2.$$ The solution $\omega^\star(\lambda)$ must obey the complementary slackness condition
\begin{equation}\label{eq:slack}
0=\lambda(1-\omega^\star(\lambda))
\end{equation}which, for $0<\omega^\star<1$, is satisfied only at $\lambda^\star=0$.

Plugging in  $\lambda^\star=0$ into equation \eqref{eq:cond}, we see that this solution is feasible only when 
\begin{equation}\label{eq:condition}
\mathcal{R}_0(1)-\tr {H_0}\Omega_{\Sigma}=\tr {H_0} \Sigma_1-\tr {H_0}\Omega_\Sigma> 0,	
\end{equation}
else the solution $\omega^\star=\omega^\star(0)\le0$, violates the constraint $\omega^\star\ge0$. Therefore, when \eqref{eq:condition} is satisfied we have $$\omega^\star=\omega^\star(0)=\frac{\mathcal{R}_0(1)-\tr {H_0}\Omega_{\Sigma}}{\mathcal{R}_0(S_1)+\mathcal{R}_0(S_2)-2\tr {H_0}\Omega_\Sigma}\equiv \frac{\tr {H_0}(\Sigma_1-\Omega_\Sigma)}{\tr {H_0}(\Sigma_1+\Sigma_2-2\Omega_\Sigma)+\tau_2^\top Q_2^\top {H_0}Q_2\tau_2},$$which yields the first claimed solution.

Consider that the condition in \eqref{eq:condition} is violated. Then, for $C= [\mathcal{R}_0(1)-\tr {H_0}\Omega_\Sigma]\le0$,
and 
\begin{flalign*}
\mathcal{R}_0(\omega)&=\mathcal{R}_0(1)-2\omega {[\mathcal{R}_0(1)-\tr {H_0}\Omega_\Sigma]}+\omega^2\left[\mathcal{R}_0(S_1)+\mathcal{R}_0(S_2)-2\tr {H_0}\Omega_\Sigma\right]\\&=\mathcal{R}_0(1)+\omega\underbrace{[-2\omega C]}_{>0}+\omega^2\underbrace{\mathcal{J}}_{>0}.	
\end{flalign*}From the above, we see that $\mathcal{R}_0(\omega)$ is minimized at $\omega^\star=0$, which yields the minimal asymptotic expected loss, and the second claimed solution.
\end{proof}

\begin{proof}[Proof of Lemma 2]
It follows directly from Lemma 3, and Lemma 1 that $\mathcal{R}_0(\widehat\omega)=\mathcal{R}_0(\omega_0)+o_p(1)$. Now, recall that 
\begin{flalign*}
\mathcal{R}_0(\omega)&=(1-\omega)^2\tr {H_0}\Sigma_1+\omega^2[\tr {H_0}\Sigma_2]+2\omega(1-\omega) \tr {H_0}\Omega_\Sigma\\&\equiv(1-\omega)^2\mathcal{R}_0(1)+\omega^2\mathcal{R}_0(2)+2\omega(1-\omega) \tr {H_0}\Omega_\Sigma.
\end{flalign*}Under $\omega=\omega_0$, and for $C_0=\mathcal{R}_0(1)+\mathcal{R}_0(2)$, we can rewrite the above as 
\begin{flalign}
\mathcal{R}_0(\omega_0)&=\frac{\mathcal{R}_0(2)^2}{C_0^2}\mathcal{R}_0(1)+\frac{\mathcal{R}_0(1)^2}{C^2_0}\mathcal{R}_0(2)+2\frac{\mathcal{R}_0(1)\mathcal{R}_0(2)}{C_0^2}\tr {H_0}\Omega_\Sigma\label{eq:naivesopt}.
\end{flalign}

Consider that $\min\{\mathcal{R}_0(1),\mathcal{R}_0(2)\}=\mathcal{R}_0(1)$, and using equation \eqref{eq:naivesopt} to rewrite $\mathcal{R}_0(w_0)-\mathcal{R}_0(1)$ as
\begin{flalign*}
\mathcal{R}_0(w_0)-\mathcal{R}_0(1)&=\frac{\mathcal{R}_0(1)}{C_0^2}\left[\mathcal{R}_0(2)^2+\mathcal{R}_0(1)\mathcal{R}_0(2)+2\mathcal{R}_0(2)\tr {H_0}\Omega_\Sigma-\{\mathcal{R}_0(1)+\mathcal{R}_0(2)\}^2\right].
\end{flalign*}
Using the definition of $C_0^2$ and completing the square, we see that
\begin{flalign}
\mathcal{R}_0(w_0)-\mathcal{R}_0(1)&=\frac{\mathcal{R}_0(1)}{C_0^2}\left[C_0^2-\mathcal{R}_0(1)^2-\mathcal{R}_0(1)\mathcal{R}_0(2)+2\mathcal{R}_0(2)\tr {H_0}\Omega_\Sigma-C_0^2\right]\nonumber\\&=-\frac{\mathcal{R}_0(1)}{C_0^2}\left[\mathcal{R}_0(1)^2+\mathcal{R}_0(1)\mathcal{R}_0(2)-2\mathcal{R}_0(2)\tr H_0\Omega_\Sigma\right].\label{eq:naivesopt2}
\end{flalign}	Since $\tr H_0\Omega_\Sigma<0$, the RHS of the above is negative and $\mathcal{R}_0(w_0)-\mathcal{R}_0(1)\le0$. 

In the case where 	$\min\{\mathcal{R}_0(1),\mathcal{R}_0(2)\}=\mathcal{R}_0(2)$, repeating the above steps yields 
\begin{flalign*}
\mathcal{R}_0(w_0)-\mathcal{R}_0(2)&=-\frac{\mathcal{R}_0(2)}{C_0^2}\left[\mathcal{R}_0(2)^2+\mathcal{R}_0(1)\mathcal{R}_0(2)-2\mathcal{R}_0(1)\tr H_0\Omega_\Sigma\right].
\end{flalign*}	Thus, a sufficient condition for $\mathcal{R}_0(w_0)-\mathcal{R}_0(2)$ is that $\left[\mathcal{R}_0(2)-2\tr H_0\Omega_\Sigma\right]\ge0$. Hence, so long as $\tr H_0\Omega_\Sigma\le \frac{1}{2}\min\{\mathcal{R}_0(1),\mathcal{R}_0(2)\}$ the result follows. 

\end{proof}

\end{document}